\documentclass[]{article}
\usepackage[left=2.2cm,top=2cm,right=2.2cm,bottom=3cm,bindingoffset=0.5cm]{geometry}
\usepackage{times}
\usepackage{url}
\usepackage[T1]{fontenc} 
\usepackage{authblk}
\usepackage[numbers, sort, comma, square]{natbib}
\PassOptionsToPackage{numbers,compress}{natbib}
\RequirePackage[title]{appendix}
\RequirePackage{xcolor,soul}
\RequirePackage[normalem]{ulem}
\RequirePackage[colorlinks,allcolors=blue!70!black]{hyperref}
\RequirePackage{amsfonts,amssymb,amsthm,amsmath,bm}
\RequirePackage{multirow,diagbox,subfig}

\RequirePackage{enumitem} \setlist{nosep}
\RequirePackage[skip=.5ex minus.2ex]{caption}
\RequirePackage{pgfplots}
\pgfplotsset{compat=newest}
\usepgfplotslibrary{groupplots}

\RequirePackage[disable,textsize=small,backgroundcolor=yellow!10]{todonotes} 
\RequirePackage{subdef}

\usepackage{stackengine}
\usepackage{algorithm}
\usepackage[noend]{algorithmic}
\newcommand{\sprot}{\sigma_{\text{prot}}}

\renewcommand{\Gcal}{G}
\renewcommand{\Vcal}{V}
\renewcommand{\Ecal}{E}
\renewcommand{\Rcal}{R}

\DeclareMathOperator{\Ncal}{nbr}
\DeclareMathOperator{\nnbr}{\overline{\Ncal}}
\renewcommand{\nbu}{\nnbr(u)}
\DeclareMathOperator{\Pcal}{prot}
\DeclareMathOperator{\Tcal}{pub}

 \renewcommand{\cite}{\citep}
\renewcommand{\Omega}{{\small \textsc{Vol}}}

\DeclareMathOperator{\Area}{{\small \textsf{ComVol}}}

\begin{document}
\def\papertitle{Differentially Private Link Prediction\\ With Protected Connections}
\title{\papertitle}
\author{Abir De}
\author{Soumen Chakrabarti}

\affil{Indian Institute of Technology Bombay,\\ \{abir, soumen\}@cse.iitb.ac.in }

\date{}
\maketitle

\begin{abstract}
Link prediction (LP) algorithms propose to each node a ranked list of
nodes that are currently non-neighbors, as the most likely candidates for future linkage. Owing to increasing concerns about privacy, users (nodes) may prefer to keep some of their connections protected or private. Motivated by this observation, our goal is to design a
\dpp\ LP algorithm, which trades off between privacy of the protected node-pairs and the link prediction accuracy. More specifically, we first propose a form of differential privacy on graphs, which models the privacy loss only of those node-pairs which are marked as protected. Next, we develop \our, a learning to rank algorithm, which applies a monotone transform to base scores from a non-private LP system, and then adds noise.  \our\ is trained with a privacy induced ranking loss, which optimizes the ranking utility for a given maximum allowed level of privacy leakage of the protected node-pairs. Under a recently-introduced latent node embedding model, we present a formal trade-off between privacy and LP utility. Extensive experiments with several real-life graphs and several LP heuristics show that \our\ can trade off between privacy and predictive performance more effectively than several alternatives.
\end{abstract}

\section{Introduction}
\label{sec:Intro}

Link prediction (LP) \citep{LibenNowellK2007LinkPred,LuZ2011LinkPred} 
is the task of predicting future edges that 
are likely to emerge between nodes in a social network, given 
historical views of it.  Predicting new research collaborators, new Facebook friends and new LinkedIn connections
are examples of LP tasks. 
LP algorithms usually
presents to each node $u$ a ranking of the most promising non-neighbors
$v$ of~$u$ at the time of prediction.

Owing to the growing concerns about privacy~\cite{p1,MachKDS2011AccuratePrivate,waniek2019hide},
social media users often prefer to mark their sensitive contacts, \eg,
dating partners, prospective employers, doctors, etc. as `protected' from 
other users and the Internet at large.
However, the social media hosting platform itself continues to enjoy
full access to the network, and may exploit them for
link prediction.  In fact, many popular LP algorithms
often leak neighbor information~\cite{waniek2019hide}, at least in a probabilistic sense,
because they take advantage of rampant 
\emph{triad completion}: when node $u$ links to $v$, very often there exists a
node $w$ such that edges $(u,w)$ and $(w,v)$ already exist.
Therefore,\todo{SC:Check} recommending $v$ to $u$  allows $u$ to learn about the edge $(w,v)$, which may result in breach of privacy of $v$ and~$w$. 
Motivated by these observations, a recent line of work~\cite{MachKDS2011AccuratePrivate,xu2018dpne}
proposes privacy preserving LP methods.
While they introduce sufficient noise to ensure a specified level of privacy, they do not directly optimize for the LP ranking utility.  Consequently, they show poor predictive performance.

\subsection{Present work}
In this work, we design a learning to rank algorithm which aims to optimize the LP utility under a given amount
of privacy leakage of the node-pairs (edges or non-edges) which are marked as protected. Specifically, we make the following contributions.

\xhdr{Modeling differential privacy of protected node pairs}
Differential privacy (DP) \citep{DworkR2014AlgoDP} is concerned with
privacy loss of a \emph{single entry} in a database upon (noisy) datum disclosure.
As emphasized by~\citet[Section~2.3.3]{DworkR2014AlgoDP},
the meaning of {\em single entry} varies across applications, especially in graph databases.
DP algorithms over graphs predominantly consider
two notions of granularity, \eg, node and edge DP, where one entry is represented by a node and an edge respectively.
While node DP provides a stronger privacy guarantee than edge DP, it is often stronger than what is needed for applications.  As a result, it
can result in unacceptably low utility~\cite{song2018differentially, DworkR2014AlgoDP}. To address this challenge,
we adopt a variant of differential privacy, called \eprivacy,
which accounts for the loss of privacy of only those node-pairs (edges or non edges), which are marked as private. %
Such a privacy model is well suited for LP in social networks,
where users may hide only some sensitive contacts and leave the remaining relations public~\cite{DworkR2014AlgoDP}. 
It is much stronger than edge DP, but weaker than node DP. However, it serves the purpose in the context
of our current problem.

\xhdr{Maximizing LP accuracy subject to privacy constraints}
Having defined an appropriate granularity of privacy, we develop \our, a perturbation wrapper around a base LP algorithm, that renders it private in the sense of node-pair protection described above.  Unlike DP disclosure of \emph{quantities}, the numeric scores from the base LP algorithm are not significant in themselves, but induce a ranking.  This demands a fresh approach to designing \our, which works in two steps.
First, it transforms the base LP scores using a monotone deep network \citep{umnn}, which ensures that the ranks remain unaffected in absence of privacy.  Noise is then added to these transformed scores, and the noised scores are used to sample a privacy-protecting ranking over candidate neighbors.
\our\ is trained using a ranking loss (AUC) surrogate.
Effectively, \our{} obtains a sampling distribution that optimizes ranking utility, while maintaining privacy for protected pairs.
Extensive experiments show that \our\ outperforms  state-of-the-art baselines \todo{@ad just these 2?} \citep{MachKDS2011AccuratePrivate,McSherryT2007Mechanism,geng2015optimal} in terms of predictive accuracy, under identical privacy constraints\footnote{Our code is available in  https://github.com/abir-de/dplp}.

\xhdr{Bounding ranking utility subject to privacy}
Finally, we theoretically analyze the ranking utility provided by three LP algorithms~\cite{LibenNowellK2007LinkPred}, viz.,  Adamic-Adar, Common Neighbors and Jaccard Coefficient, in presence of
privacy constraints.  Even without privacy constraints, these and other LP algorithms are heuristic in nature and can produce imperfect rankings.
Therefore, unlike \citet{MachKDS2011AccuratePrivate},
we model utility not with reference to a non-DP LP algorithm, but a generative process that creates the graph.
Specifically, we consider the generative model proposed by~\citet{SarkarCM2011LPembed},
which establishes predictive power of popular (non-DP) LP heuristics.  They model latent node embeddings in a geometric
space \citep{HoffRH2002LatentSpace} as causative forces that
drive edge creation {but without any consideration for privacy.  In contrast, we bound the \emph{additional loss} in ranking utility when privacy constraints are enforced and therefore, our results change the setting and require new techniques.}

\subsection{Related work}
The quest for correct privacy protection granularity has driven much work in the DP community. \citet{kearns2015privacy} propose a model in which a fraction of nodes are completely protected, whereas the others do not get any privacy assurance. This fits certain applications \eg, epidemic or terrorist tracing, but not LP, where a node/user may choose to conceal only a few sensitive contacts (\eg, doctor or partner) and leave others public.  Edge/non-edge protection constraints can be encoded in the formalism of group \cite{dwork2011differential}, Pufferfish \cite{song2017pufferfish,kifer2014pufferfish,song2017composition} and Blowfish \cite{he2014blowfish} privacy. 
However, these do not focus on LP tasks on graphs.
\citet{ppsRavi} partition a graph into private and public subgraphs, but for designing efficient sketching and sampling algorithms, not directly related to \dpp\ LP objectives.
{In particular, they do not use any information from private graphs; however, a differentially private algorithm does use private information after adding noise into it. Moreover, their algorithms do not include any learning to rank method for link prediction.}
\citet{abadi2016deep} design a DP gradient descent algorithm for deep learning, which however,
 does not aim to optimize the associated utility. Moreover, they work with pointwise loss function, whereas our objective is a pairwise loss.
Other work \citep{geng2015optimal, geng2018optimal,ghosh2012universally} aims to find optimal noise distribution for queries seeking real values.  
However, they predominantly use variants of noise power in their objective, whereas we use ranking loss, designed specifically for LP tasks. 


\section{Preliminaries}
\label{sec:Prelim}

In this section, we first set up the necessary notation and give a brief overview 
of a generic non-private LP algorithm on a social network. Then we introduce the notion of protected and non-protected node pairs.
\subsection{Social network and LP algorithm}
We consider a snapshot of the social network as an undirected
graph $\Gcal= (\Vcal, \Ecal)$ with vertices $\Vcal$ and edges $\Ecal$.
Neighbors and non-neighbors of node $u$ are denoted as the sets $\Ncal(u)$ and  $\nbu$, respectively.
Given this snapshot, a non-private LP algorithm $\Acal$
first implements a scoring function $s_\Acal(u,\cdot): \nbu\to\RR^+$
and then sorts the scores $s_\Acal(u,v)$ of all the non-neighbors $v\in \nbu$ in decreasing order to obtain a ranking  of the potential neighbors for each query node $u$.
Thus, the LP algorithm $\Acal$ provides a map:
$\pi^\Acal_u:\{1,2,\hdots,|\nbu|\}  \leftrightarrow_{1:1} \nbu$.
Here, $\pi^\Acal_u (i)$ represents the node at rank $i$, recommended to node $u$ by algorithm $\Acal$.
For brevity, we sometimes denote $u^\Acal _i = \pi^\Acal_u (i)$.
Note that, although $\pi^\Acal_u$ can provide ranking of all candidate nodes, most LP applications consider only top $K$ candidates for some small value of $K$.  When clear from context, we will use $\pi^\Acal_u$ to denote the top-$K$ items.

\subsection{Protected and non-protected node-pairs} 
We assume that each node $u\in\Vcal$ partitions all other nodes $\Vcal\cp \set{u}$
into a protected node set $\Pcal(u)$ and a non-protected/public node set $\Tcal(u)$.  We call
$(u,w)$ a protected node-pair if either $w$ marks $u$ as protected, \ie, $u\in \Pcal(w)$ or vice-versa, \ie, $w\in\Pcal(u)$. 
In general, any node $u$ can access three sets of information:
(i)~its own connections with other nodes, \ie, $\Ncal(u)$ and $\nnbr(u)$;
(ii)~the possible connections between any node $w$ (not necessarily a neighbor) with nodes $\Tcal(w)$ not protected by~$w$; and, 
(iii)~the ranked list of nodes recommended to itself by any LP algorithm.
Note that a node-pair $(u,w)$ is accessible to $u$ even if $u\in \Pcal(w)$, since $u$ knows its own edges (and non-edges). On the other hand, if a third node $v\in \Vcal\cp\set{u,w}$ is marked as protected by $w$, \ie, $v\in \Pcal(w)$, then $u$ should not know whether $v\in \Ncal(w)$ or $v\in \nnbr(w)$.  However, when a non-private LP algorithm recommends nodes to $u$, $u$ can exploit them to reverse-engineer such protected links. For example, if $w$ is connected to both $u$ and $v$, then a triad-completion based LP heuristic (Adamic-Adar or Common-Neighbor) is likely to recommend $v$ to the node $u$, which in turn allows $u$ to guess the presence of the edge $(v,w)$. Figure~\ref{fig:priv-illus} illustrates our problem setup.
\begin{figure}[t]
\centering
\includegraphics[width=0.51\textwidth]{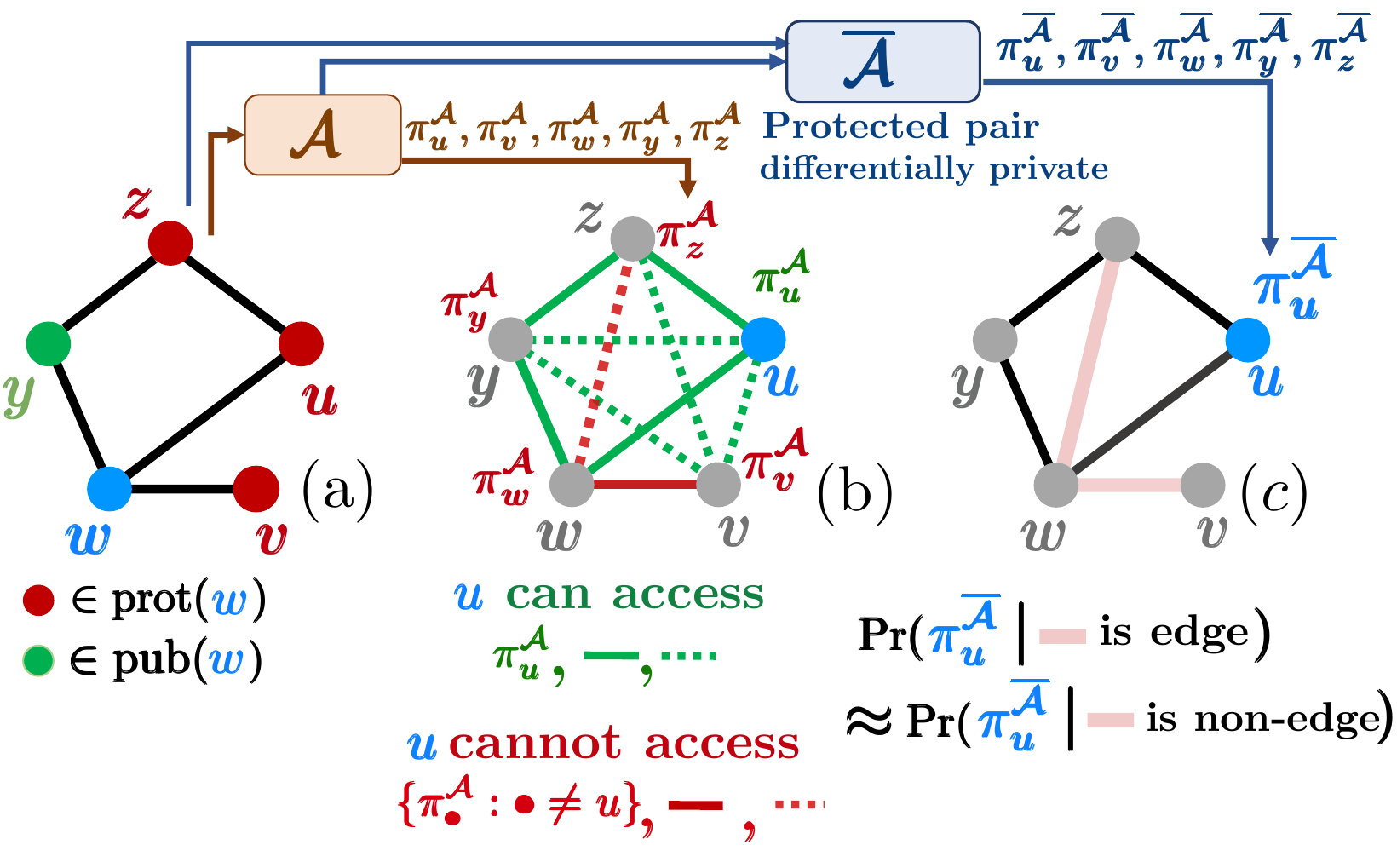}
\caption{Problem setup.
(a)~Nodes in $\Pcal(w)$ colored red; nodes in $\Tcal(w)$ colored green.
(b)~$u$ cannot see red edges (solid), red non-edges (dotted) and recommendations ($\pi^{\Acal} _{\bullet}, \bullet \neq u$) to other nodes.
$u$ can see green edges (solid), green non-edges (dotted) and recommendation ($\pi^{\Acal} _{u}$) to itself. 
(c)~Output of $\App$ should be insensitive to protected pairs.}
\label{fig:priv-illus}
\end{figure}

\section{Privacy model for protected pairs} 

In this section, we first present our privacy model which assesses the privacy loss of protected node-pairs, given the predictions made by an LP algorithm. Then we introduce the sensitivity function, specifically designed for graphs with protected pairs. 
\subsection{Protected-pair differential privacy}
In order to the prevent privacy leakage of the protected connections,
we aim to design a privacy preserving LP algorithm $\App$ based 
on $\Acal$, meaning that $\App$ uses the scores $s_{\Acal}(u,\cdot)$ to present to each node $u$ a ranking $\pi^\App_u$, while maintaining the privacy of all the protected node pairs in the graph, where one of the two nodes in such a pair is not $u$ itself.
We expect that $u$ has strong reasons to link to the top nodes in $\pi^\App_u$, and that the same list will occur again with nearly equal probability, if one changes (adds or deletes) edges between any other node $w$ and the protected nodes of $w$, other than $u$, \ie, $\Pcal(w)\cp\set{u}$.  We formally define such changes in the graph in the following.

\begin{definition}[\bfseries $u$-neighboring graphs] \label{def:uNbrGraphs}
Given two graphs $\Gcal=(\Vcal,\Ecal)$ and $\Gcal'=(\Vcal,\Ecal')$ with the same node set $\Vcal$, and a node $u\in\Vcal$,  we call $\Gcal$, $\Gcal'$ as $u$-neighboring graphs, if there exists some node $w\ne u$ (not necessarily a neighbor), such that
\begin{enumerate}
\item protected nodes $\Pcal(w)$ are the same in $\Gcal$ and~$\Gcal'$,
\item public nodes $\Tcal(w)$ are the same in $\Gcal$ and~$\Gcal'$, and
\item $\Ecal'{\setminus}\{(w,v)\given v{\in}\Pcal(w)\cp u\} {=} \Ecal
{\setminus} \{(w,v)\given v{\in}\Pcal(w)\cp u \}$.
\end{enumerate}
The set of all $u$-neighboring graph pairs is denoted as $\mathfrak{N}_u$.
\end{definition}

\noindent In words, the last requirement is that, for the purpose of recommendation to $u$, the information not protected by $w$ is the same in $\Gcal$ and~$\Gcal'$.
$u$-neighboring graphs may differ only in the protected pairs incident to
some other node $w$, except $(w,u)$.  Next, we propose our privacy model which characterizes the privacy leakage of protected node-pairs in the context of link prediction.

\begin{definition}[\bfseries Protected-pair differential privacy] \label{def:epp}
We are given a non-private base LP algorithm $\Acal$, and a randomized LP routine $\App$ based on $\Acal$, which returns a ranking of $K$ nodes.  We say that ``$\App$ ensures $\epsilon_p$-\eprivacy'', if, for all nodes $u\in\Vcal$ and for all ranking $\Rcal_K$ on the candidate nodes, truncated after position~$K$, we have $\Pr(\pi^\App _u =\Rcal_K\given \Acal, \Gcal) \le e^{\epsilon_p} \Pr(\pi^\App _u =\Rcal_K \given \Acal, \Gcal')$, whenever $\Gcal$ and $\Gcal'$ are $u$-neighboring graphs.  
\end{definition}

\noindent In words, $\App$ promises to recommend the same list to $u$, with nearly the same probability, even if we arbitrarily change the protected partners of any node $w$, except for $u$ itself, as that information is already known to~$u$.
Connections between protected-pair differential privacy and earlier privacy definitions are drawn in Appendix~\ref{sec:connectionAppen}.

\subsection{Graph sensitivity}

Recall that $\Acal$ recommends $v$ to $u$ based on a score $s_\Acal(u,v)$.  A key component of our strategy is to apply a (monotone) transform $f(s_\Acal(u,v))$ to balance it more effectively against samples from a noise generator.  To that end, given base LP algorithm $\Acal$, we define  the sensitivity of any function on the scores $\tf(s_{\Acal}(u,\cdot)):\Vcal\to\RR^+$, which  will be subsequently used to design any \eprivate\ LP algorithm.

\begin{definition}[\bfseries Sensitivity of $\tf \circ s_\Acal$]
The sensitivity $\Deltaf$ of a transformation $\tf$ on the scoring function $s_\Acal$ is defined as
\begin{align}
\Deltaf = \max_u \max_{(G,G')\in\mathfrak{N}_u} \max_v 
\big|&\tf(s_\Acal(u,v|\Gcal))   
 - \tf(s_\Acal(u,v|\Gcal'))\big|
\end{align}
Recall $\mathfrak{N}_u$ represents all $u$-neighboring graph pairs.
\end{definition}

\section{Design of privacy protocols for LP}
\label{sec:Priv}

In this section, we first state our problem of designing a privacy preserving LP algorithm $\App$--- based on a non-private LP algorithm $\Acal$--- which aims to provide an optimal trade-off between privacy leakage of the protected node pairs and LP accuracy.
Next,  we develop \our, our proposed algorithm which solves this problem,

\subsection{Problem statement} 

Privacy can always be trivially protected by ignoring private data and sampling each node uniformly, but such an LP `algorithm' has no predictive value. Hence, it is important to design a privacy preserving LP algorithm which can also provide a reasonable LP accuracy. To this end, given a base LP algorithm $\Acal$, our goal is to design a \eprivate\ algorithm $\App$ \emph{based on $\Acal$}, which maximizes the link prediction accuracy.  $\App$ uses random noise, and induces a distribution $\pi^\App_u$ over rankings presented to~$u$.  Specifically, our goal is solve the following optimization problem:
\begin{align}
&\underset{\App}{\text{maximize}}\  {\sum_{u\in\Vcal}\EE_{\Rcal_K \sim \PP( \pi^\App _u |\Acal, \Gcal)}  \text{AUC}(\Rcal_K)}
\label{eq:prob-def}\\
& \text{s.t.} \quad \App \quad \text{is $K\epsilon_p$ \eprivate}. 
\notag
\end{align}
Here, AUC($\Rcal_K$) measures the area under the ROC curve given by $R_K$--- the ranked list of recommended nodes truncated at~$K$.
Hence, the above optimization problem maximizes the expected AUC over the rankings generated by the randomized protocol $\App$, based on the scores produced by~$\Acal$.
The constraint indicates that  predicting \emph{a single link} using $\App$ requires to satisfy $\epsilon_p$-\eprivacy.

\subsection{\our: A \eprivate\ LP algorithm}

A key question is whether well-known generic privacy protocols like Laplace or exponential noising are adequate solvers of the optimization problem in Eq.~\eqref{eq:prob-def}, or might a solution tailored to \eprivacy{} in LP do better.  We answer in the affirmative by presenting \our, a novel DP link predictor, with  five salient components.
\begin{enumerate}
\item Monotone transformation of the original scores from~$\Acal$.
\item Design of a ranking distribution $\pi^\App_u$, which suitably introduces noise to the original scores to make a sampled ranking insensitive to protected node pairs.
\item Re-parameterization of the ranking distribution to facilitate effective training.
\item Design of a tractable ranking objective that approximates AUC in Eq.~\eqref{eq:prob-def}.
\item Sampling current non-neighbors for recommendation. 
\end{enumerate}

\paragraph{Monotone transformation of base scores}
We equip $\App$ with a trainable monotone transformation $f:\RR^+ \to \RR^+$, which is applied to base scores $s_{\Acal}(u,\cdot)$.  Transform $f$ helps us more effectively balance signal from $\Acal$ with the privacy-assuring noise distribution we use.  We view this transformation from two perspectives. 

First, $f(s_{\Acal}(u,\cdot))$ can be seen as a new scoring function which aims to boost the ranking utility of the original scores $s_{\Acal}(u,\cdot)$, under a given privacy constraint. In the absence of privacy constraints, we can reasonably demand that the ranking quality of $\tf\circ s_\Acal$ be at least as good as~$s_\Acal$.
In other words, we expect that the transformed and
the original scores provide the same ranking utility, i.e., the same list of recommended nodes, in the absence of privacy constraints.
A monotone $f$ ensures such a property.

Second, as we shall see, $f$ can also be seen as a \emph{probability distribution transformer}.  Given a sufficiently expressive neural representation of $f$ (as designed in Section~\ref{sec:dplp:neural-f}), any arbitrary distribution on the original scores can be fitted with a tractable known distribution on the transformed scores.  This helps us build a high-capacity but tractable ranking distribution $\PP(\pi^\App _u \given \Acal, \Gcal)$, as described below.

\paragraph{Design of a ranking distribution} 
Having transformed the scores $s_{\Acal}(u,.)$ into $\tf(s_{\Acal}(u,v))$, we draw nodes $v$ using an iterative exponential mechanism on the transformed scores $\tf(s_{\Acal}(u,v))$. 
Specifically, we draw the first node $v$ with probability proportional to $\exp \left(  {\epsilon_p \tf(s_{\Acal} (u,v))}/2{\Deltaf}\right)$, then repeat for collecting $K$ nodes in all.  Thus, our ranking distribution in Eq.~\eqref{eq:prob-def} is: 
\begin{align}
\Pr(\pi^\App _u &  = \Rcal_K |\Acal, \Gcal) =  
 \displaystyle \prod_{j=1} ^ K \dfrac{\exp\left( \dfrac{\epsilon_p \tf(s_\Acal(u, \Rcal_K(j)))}{2\Deltaf} \right)}{{\sum_{w\not \in \Rcal_K(1,..,j-1)}}\exp\left(\dfrac{\epsilon_p \tf(s_\Acal(u,w))}{2\Deltaf} \right) },\label{eq:exp-mech-iterative}
\end{align}
In the absence of privacy concerns ($\epsilon_p\to\infty$), the candidate node with the highest score $\tf(s_{\Acal}(u,.))$  gets selected in each step, and the algorithm reduces to the original non-private LP algorithm $\Acal$, thanks to the monotonicity of $\tf$. On the other hand, if $\epsilon_p\to 0$, then every node has an equal chance of getting selected, which preserves privacy but has low predictive utility.

Indeed, our ranking distribution in Eq.~\eqref{eq:exp-mech-iterative} follows an exponential mechanism on the new scores $\tf(s_{\Acal}(u,\cdot))$. However, in principle, such a mechanism can capture any arbitrary ranking distribution $\Pr(\pi^\App _u|\Acal, \Gcal)$, given sufficient training and expressiveness of $\tf$. 

Finally, we note that $\tf$ can be designed to limit~$\Deltaf$.  E.g., if we can ensure that $\tf(\cdot)$ is positive and bounded by $B$, then we can have $\Deltaf \le B$; if we can ensure that the derivative $\tf'(\cdot)$ is bounded by $B'$, then, by the Lipschitz property, $\Deltaf\le B'\Delta_{\Acal}$.

\paragraph{Re-parameterization of the ranking distribution}
In practice, the above-described procedure for sampling a top-$K$ ranking results in high estimation variance during training~\cite{zhao2011analysis, policygrad1, policygrad2, policygrad3}.  To overcome this problem, we  make a slight modification to the softmax method.  We first add i.i.d Gumbel noise $\eta_{u,v} \sim \gumbel(2\Deltaf/\epsilon_p)$ to each score $\tf(s_{\Acal}(u,v))$ and then take top $K$ nodes 
with highest noisy score~\cite[Lemma~4.2]{durfee2019practical}. Such a distribution allows an easy re-parameterization trick --- $\gumbel(2\Deltaf/\epsilon_p)= (2\Deltaf/\epsilon_p)\cdot\gumbel(1)$ --- which reduces the parameter variance.
With such a re-parameterization of the softmax ranking distribution in Eq.~\eqref{eq:exp-mech-iterative}, we get the following alternative representation of our objective in Eq.~\eqref{eq:prob-def}:
\begin{align}
{ {\sum_{u\in\Vcal}}\underset{\set{\eta_{u,\bullet}}}{\EE}\underset{\sim \gumbel(1)}{} 
     \text{AUC}\left(\Rcal_K\bigg(\bigg\{\tf(s_{\Acal} (u,\bullet) + \tfrac{2\Deltaf \eta_{u,\bullet}}{\epsilon_p}\bigg\}\bigg) \right)},\nn
\end{align}
where $\Rcal_K(.)$ gives a top-$K$ ranking based on the noisy transformed scores $\set{\tf(s_{\Acal} (u,\bullet) + {2\Deltaf \eta_{u,\bullet}}/{\epsilon_p}}$ over the candidate nodes for recommendation to~$u$.

\paragraph{Designing a tractable ranking objective}
Despite the aforementioned re-parameterization trick, 
standard optimization tools face several challenges to maximize the above objective. First, optimizing AUC is an NP-hard problem.
Second, truncating the list up to $K$ nodes often pushes many neighbors out of the list at the initial stages of learning, which can lead to poor training~\cite{Liu2009LearningToRank}. 
To overcome these limitations, we replace AUC \cite{auc1}  with a surrogate pairwise hinge loss $\ell(\cdot)$ and consider a non-truncated list $\Rcal_{|\Vcal|-1}$ of all possible candidate nodes \emph{during training}~\cite{auc1,auc3}.
Hence our optimization problem becomes:
\begin{align}
     \underset{\tf}{\text{min}} \  {\sum_{u\in\Vcal}}
    \   {\EE}_{\set{\eta_{u,\bullet}} \sim \gumbel(1) }{}  \ell\left( \left\{\tf(s_{\Acal}(u,\bullet)), \eta_{u\bullet}\right\}; \Gcal \right) \label{eq:opt-problem}
\end{align}
where $  \ell\left( \left\{\tf(s_{\Acal}(u,\bullet)), \eta_{u\bullet}\right\}; \Gcal \right)$ is given by a pairwise ranking loss surrogate over the noisy transformed scores:
\begin{align}
\sum_{\substack{g\in \Ncal(u)\cap \Tcal(u)\\ b\in \nbu \cap \Tcal(u) }}  
   \hspace{-2em}  \text{ReLU} \bigg[\varrho + \tf(s_{\Acal} (u,b)) + \dfrac{2\Deltaf\eta_{ub}}{\epsilon_p}  
 -\tf(s_{\Acal} (u,g))-\dfrac{2\Deltaf\eta_{ug}}{\epsilon_p}  \bigg].  \label{eq:ranking-loss}
\end{align}
The above loss function encourages that the (noisy) transformed score of a neighbor $g$ exceeds the corresponding score of a non-neighbor $b$ by at least (a tunable) margin~$\varrho$.
In absence of privacy concern, \ie, $\epsilon_p\to\infty$, it
is simply $\text{ReLU}\left[\varrho + \tf(s_{\Acal} (u,b))-\tf(s_{\Acal} (u,g))\right]$, which would provide the same ranking as $\Acal$, due to the monotonicity of $f$.

We would like to point out that the objective in Eq~\ref{eq:ranking-loss} uses only non-protected pairs $\{(u,v)|v\in \Tcal(u)\}$ to train~$\tf$.  This ensures no privacy is leaked during training.

\setlength{\textfloatsep}{3pt}
\begin{algorithm}[t]
\small
\caption{\our: Learns $\App$ based on a non-private  LP algorithm $\Acal$ and then uses it to recommend $K$ nodes to~$u$.}
 \label{alg:dpone}
  \begin{algorithmic}[1]
    \STATE \textbf{Input} $\Gcal=(\Vcal,\Ecal)$; protected and non-protected node-sets $\Pcal(\cdot), \Tcal(\cdot)$; scores $s_{\Acal}(.,.)$ given by $\Acal$; privacy leakage $\epsilon_p$, margin $\varrho$.
    \STATE \textbf{Output} $\Rcal_K:$ List of top-$K$ recommended nodes to node $u$   
    \STATE \textbf{Initialize} $\Rcal_K(1 \ldots K)\leftarrow \emptyset$
    \STATE $\tf\leftarrow$ \textsc{Train}$\left(\Gcal,\set{s_{\Acal}(w,v)\given v \in \Tcal(w), u\in\Vcal},\epsilon_p, \varrho\right)$
     \STATE $\text{candidates}\leftarrow \nnbr(u)$
     \FOR{$j$ in $1\ldots K$}
     \STATE $w\sim \textsc{SoftMax}\left(\set{\epsilon_p\tf(s_{\Acal}(u,v)/2\Deltaf \given v\in\text{candidates} } \right)$
     \STATE $\text{candidates}\leftarrow \text{candidates}\cp \set{w}$
     \STATE $\Rcal_K(j)\leftarrow w$
     \ENDFOR
     \STATE \mbox{\bf Return} $\Rcal_K$
  \end{algorithmic}
\end{algorithm}
\xhdr{Sampling nodes for recommendation} Once we train $\tf$ by solving the optimization problem in Eq.~\eqref{eq:opt-problem}, we draw top-$K$ candidate nodes using the trained $\tf$.  Specifically, we first sample a non-neighbor $v$ with probability proportional to
$\exp\left(\epsilon_p \tf(s_\Acal(u, v))/2\Deltaf\right)$
and then repeat the same process on the remaining non-neighbors $\nbu\cp\{v\}$ and continue $K$ times.

\xhdr{Overall algorithm}
We summarize the overall algorithm in Algorithm~\ref{alg:dpone}, where the $\textsc{Train}(\bullet)$ routine solves the
optimization problem in Eq.~\eqref{eq:opt-problem}. Since there is no privacy leakage during training, the entire process --- from training to test --- ensures $K\epsilon_p$ \pprivacy, which is formalized in the following proposition and proven in Appendix~\ref{sec:appen:dponeX}.
\begin{proposition}\label{thm:dpone} 
Algorithm~\ref{alg:dpone} implements $K\epsilon_p$-\pprivate{} LP.
\end{proposition}
\paragraph{Incompatibility with node and edge DP}
Note that, our proposal is incompatible with node or edge privacy (where each node or edge is protected).  In those cases, objective in Eq.~\eqref{eq:ranking-loss} cannot access any data in the training graph.
However, our proposal is useful in many applications like online social networks, in which,  a user
usually protects a fraction of his sensitive connections, while leaving others public. 

\subsection{Structure of the monotonic transformation}
\label{sec:dplp:neural-f}

We approximate the required nonlinear function $\tf$ with  $\tf_{\thetab}$, implemented as a monotonic neural network with parameters~$\thetab$. Initial experiments suggested that raising the raw score $s_{\Acal}(u,v)$ to a power $a>0$ provided a flexible trade-off between privacy and utility over diverse graphs.  We provide a basis set $\{a_i>0\}$ of fixed powers and obtain a positive linear combination over raw scores $s_{\Acal}(u,\cdot)$ raised to these powers.  More specifically, we compute:
 \begin{align}
    \nu_{\betab}(s_{\Acal}(u,v)) = \sum_{i=1} ^ {n_{a}} e^{\tau \beta_i} (s_{\Acal}(u,v))^{a_i},\label{eq:dplp-linear}
\end{align}
where $(a_i)_{i\in[n_a]}>0$ are fixed a~priori.
$\betab=(\beta_i)_{i\in [n_a]}$ are trainable parameters and $\tau$ is a temperature hyperparameter.  This gave more stable gradient updates than letting $a$ float as a variable.

While using $\nu_{\betab}(s_\Acal(u,v))$ as $f(s_\Acal(u,v))$ is already an improvement upon raw score $s_\Acal(u,v)$, we propose a further refinement: we  feed $\nu_{\betab}(s_{\Acal}(u,v))$ to a Unconstrained Monotone Neural Network (UMNN) \citep{umnn} to finally obtain $\tf_{\thetab}(s_{\Acal}(u,v))$.  UMNN is {parameterized using the integral of a positive nonlinear~function:}
\begin{align}
   \tf_{\thetab}(s_{\Acal}(u,v)) = \int_0 ^{ \nu_{\betab}(s_{\Acal}(u,v))} g_{\bm{\phi}}(s)ds + b_0, ~\label{eq:umnn}
\end{align}
where $g_{\bm{\phi}}(.)$ is a positive neural network, parameterized by~$\bm{\phi}$.  UMNN offers an unconstrained, highly expressive parameterization of monotonic functions, since the monotonicity herein is achieved by only enforcing a positive integrand $g_{\bm{\phi}}(\cdot)$.  In theory, with sufficient training and capacity, a `universal' monotone transformation may absorb $\nu_\beta(\cdot)$ into its input stage, but, in practice, we found that omitting $\nu_\beta(\cdot)$ and passing $s_\Acal(u,v)$ directly to UMNN gave poorer results (Appendix~\ref{sec:app:add}).
Therefore, we explicitly provide $\nu_{\betab}(s_\Acal(u,v))$ as an input to it, instead of the score $s_{\Acal}(u,v)$ itself. 
Overall, we have $\thetab=\{\betab,\bm{\phi},b_0\}$ as the trainable parameters.

\section{Privacy-accuracy trade-off of \our}
\label{sec:Qual}

It is desirable that a privacy preserving LP algorithm should be able to hide the protected node-pairs as well as give
high predictive accuracy. However, in principle, LP accuracy is likely to deteriorate with higher privacy constraint (lower $\epsilon_p$).
For example, a uniform ranking distribution, \ie, $\epsilon_p{\to} 0$ provides extreme privacy. However, it would lead to a poor prediction.  On the other hand, in absence of privacy, \ie, when $\epsilon_p\to\infty$, the algorithm enjoys an unrestricted use of the protected pairs, which  is likely to boost its accuracy. 
In this section, we analyze this trade-off from two perspectives.
\begin{enumerate}
    \item \textbf{Relative trade-off analysis.} Here, we assess the \emph{relative} loss of prediction quality of
    $\App$ with respect to the observed scores of the the base LP algorithm $\Acal$. 
    \item \textbf{Absolute trade-off analysis.} Here, we assess the loss of \emph{absolute} prediction quality of $\App$ with respect to the true utility provided by a latent graph generative process.
\end{enumerate}

\subsection{Relative trade-off analysis}

Like other privacy preserving algorithms, \our\ also introduces some randomness to the base LP scores, and therefore, any \our\ algorithm $\App$ may reduce the predictive quality of~$\Acal$.
To formally investigate the loss, we quantify the loss in scores
suffered by $\App$, compared to $\Acal$:
\begin{align}
  \util= \EE_{\App}\bigg[
  \sum_{i\in[K]} s_{\Acal} (u,\uaa_i|\Gcal) - \sum_{i\in[K]} s_{\Acal} (u,\upv _i|\Gcal)
  \bigg]\nonumber.
\end{align}
Recall  that $u^\Acal _i = \pi^\Acal_u (i)$ and $u^\App _i = \pi^\App _u (i)$.
We do not take absolute difference because the first term cannot be
smaller than the second.  The expectation is over randomness introduced
by the algorithm $\App$. The following theorem bounds $\util$ in two extreme cases (Proven in Appendix~\ref{sec:main-thm-proof}).
\begin{proposition}
\label{prop:relative-tradeoff}
If $\epsilon_p$ is the privacy parameter and $\kappa_{u,\Acal}:=\max_{i\in [|\Vcal|-1]} \set{s_{\Acal}(u,\uaa_i)- s_{\Acal}(u,\uaa_{i+1})}$, then we have:
   \begin{align}
       &\lim_{\epsilon_p\to \infty}\util = 0  \text{ and, }
        \lim_{\epsilon_p\to 0 }\util \le   K \kappa_{u,\Acal}. \nn
   \end{align}
\end{proposition}
\noindent Here, $\kappa_{u,\Acal}$ gives maximum difference between the scores of two consecutive nodes in the ranking provided by $\Acal$
to $u$.

\subsection{Absolute trade-off analysis}

Usually, $\Acal$ provides only imperfect predictions.
Hence, the above trade-off analysis which probes into the relative cost of privacy with respect to $\Acal$, may not always reflect the real effect of privacy on the predictive quality. To fill this gap, we need to analyze the \emph{absolute} prediction quality --- utility in terms of the true latent generative process --- which is harder to analyze.  Even for a non-private LP algorithms $\Acal$, absolute quality is rarely analyzed, except for \citet{SarkarCM2011LPembed}, which we discuss below.

\paragraph{Latent space network model}
We consider the latent space model proposed by \citet{SarkarCM2011LPembed}, which showed  
why some popular LP methods succeed.  In their model, the nodes lie within a $D$-dimensional hypersphere having unit volume.
Each node $u$ has a co-ordinate $\xb_u$ within the sphere, drawn uniformly at random.
Given a parameter $r>0$, the latent generative rule is that nodes $u$ and $v$ get connected if the distance
$d_{uv}=||\xb_u-\xb_v||_2 < r$.


\paragraph{Relating ranking loss to absolute prediction error}
If an oracle link predictor could know the underlying generative process, then its ranking $\pi^* _u$ would have sorted the nodes in the increasing order of distances $d_{u\bullet}$. An imperfect non-private LP algorithm $\Acal$ does not have access to these distances and therefore, it would suffer from some prediction error. 
$\App$, based on $\Acal$, will incur an additional prediction error due to its privacy preserving randomization.
We quantify this absolute loss by extending our loss function in Eq.~\eqref{eq:opt-problem}:
\begin{align}
\rlp: =
\textstyle{\sum_{{i<j\le K  }} [d_{uu^\App _i} - d_{uu^\App _j} ]_{+}.}
\end{align}
Here, recall that $u^\App _i = \pi^\App_u (i)$.
Analyzing the above loss for any general base LP method, especially deep embedding methods like GCN, is extremely challenging. Here we consider simple triad based models \eg, Adamic Adar (AA), Jaccard coefficients (JC) and common neighbors (CN). For such base LP methods,  we bound this loss in the following theorem (proven in Appendix~\ref{sec:main-thm-proof}).
\begin{theorem}\label{thm:util-main-thm-dp}
   Given $\epsilon=\sqrt{\frac {2\log (2/\delta)}{|\Vcal|}}+\frac{7 \log (2/\delta)}{3(|\Vcal|-1)}$. For $\Acal\in\set{\tAA,\tCN,\tJC}$, with probability $1-4K^2\delta$:
  \begin{align}
      \EE_{{\App}}\left[\rlp\right]= O\left(\big[2K\epsilon+ {\util }/{|\Vcal|}\big]^{\frac{1}{KD}} \right), \label{eq:abs-bound}
  \end{align}
 where $D$ is the dimension of the underlying hypersphere for the latent space random graph model.
\end{theorem}
The $2K\epsilon$ term captures the predictive loss due to imperfect LP algorithm $\Acal$ and $\util/|\Vcal|$ characterizes the excess loss from privacy constraints.  Note that $\epsilon$ here is unrelated to privacy level~$\epsilon_p$; it characterizes the upper bounds of the expected losses due to $\Acal$, and it depends on $\delta$ that quantifies the underlying confidence interval. The expectation is taken over random choices made by \dpp, and not randomness in the process generating the data.

\section{Experiments}
\label{sec:Expt}

\begin{table*}[t]
\centering
\resizebox{0.91\textwidth}{!}{
\begin{tabular}{p{1.3cm}||l|ccccc|c|ccccc}
\hline
& \multicolumn{6}{c|}{\textbf{AA}} & \multicolumn{6}{c}{\textbf{CN}} \\
&$\Acal$  & \our & {\our-\text{Lin}} & Staircase  & Lapl. & Exp. 
&$\Acal$  & \our & {\our-\text{Lin}} & Staircase  & Lapl. & Exp.  \\ \hline\hline
 

Facebook  & 0.842       
& \textbf{0.788}        &0.211  &0.163  &0.168  &0.170 & 0.827
& \textbf{0.768}        &0.544  &0.169  &0.181  &0.177
\\ \hline
USAir    & 0.899
& \textbf{0.825}        &0.498  &0.434  &0.461  &0.423 & 0.873
& \textbf{0.819}        &0.670  &0.432  &0.482  &0.435
\\ \hline
Twitter   & 0.726     
& \textbf{0.674}        &0.517  &0.504  &0.515  &0.488 & 0.718
& \textbf{0.667}        &0.630  &0.505  &0.523  &0.493
\\ \hline
Yeast  & 0.778  
& \textbf{0.696}        &0.179  &0.154  &0.154  &0.149 & 0.764
& \textbf{0.667}        &0.359  &0.146  &0.160  &0.151
\\ \hline
PB     & 0.627 
& \textbf{0.558}        &0.287  &0.254  &0.263  &0.255 & 0.593
& \textbf{0.537}        &0.389  &0.277  &0.272  &0.257
\\ \hline
& \multicolumn{6}{c|}{\textbf{GCN}} & \multicolumn{6}{c}{\textbf{Node2Vec}} \\
&$\Acal$  & \our & {\our-\text{Lin}} & Staircase  & Lapl. & Exp. 
&$\Acal$  & \our & {\our-\text{Lin}} & Staircase  & Lapl. & Exp.  \\ \hline\hline

Facebook    & 0.463     
& {0.151}        &0.154  &\textbf{0.159}  &0.145  &0.150 & 0.681
& \textbf{0.668}        &0.202  &0.157  &0.152  &0.155
\\ \hline
USAir    & 0.483
& \textbf{0.446}        &0.391  &0.389  &0.399  &0.420 & 0.752
& \textbf{0.681}        &0.424  &0.403  &0.422  &0.447
\\ \hline
Twitter  & 0.504
&  {0.496}        &0.498  &\textbf{0.516}  &0.509  &0.505 & 0.608
& \textbf{0.555}        &0.511  &0.485  &0.512  &0.505
\\ \hline
Yeast    & 0.369 
& \textbf{0.163}        &0.131  &0.116  &0.125  &0.112 & 0.836
& \textbf{0.803}        &0.180  &0.151  &0.159  &0.145
\\ \hline
PB & 0.503      
& {0.287}        &\textbf{0.293}  &0.262  &0.286  &0.261 & 0.525
& \textbf{0.473}        &0.304  &0.266  &0.282  &0.255
\\ \hline
\end{tabular} }
\caption{Performance (AUC)
for base non-private LP algorithm $\Acal$ and five different DP protocols $\App$, viz., \our, \our-Lin,
Staircase~\cite{geng2015optimal},
Laplace~\cite{MachKDS2011AccuratePrivate} and 
Exponential~\cite{MachKDS2011AccuratePrivate}
on 20\% held-out set.  Four base LP algorithms $\Acal$ \ie, \tAA, \tCN, \text{GCN}, \text{Node2Vec}, across all five datasets.
We set  the fraction of protected pairs $\sprot=0.3$, 
the privacy leakage $\epsilon_p=0.1$ 
and the number of recommended nodes $K=30$. 
Baseline LP algorithms AA and CN are triad-based.
Baselines GCN and Node2Vec are based on node embeddings.
Bold numbers correspond to the best performing $\App$, for the given base algorithm~$\Acal$.
\our\ outperforms all other candidates of $\App$ for all LP algorithms $\Acal$, except for GCN, in which, Staircase 
and \our-Lin also outperform \our\ in  several datasets. 
}
\label{tab:FixEpsFindMap}
\end{table*}

We report on experiments with five
real world datasets to show that \our\ can trade off privacy and
the predictive accuracy more effectively than three standard
DP protocols~\citep{geng2015optimal,MachKDS2011AccuratePrivate,McSherryT2007Mechanism}.

\subsection{Experimental setup}
\paragraph{Datasets}
We use five networks:
Facebook~\cite{leskovec2012learning},
USAir~\cite{usair},
Twitter~\cite{leskovec2012learning},
Yeast~\cite{von2002comparative},
PB~\cite{ackland2005mapping}.
Appendix~\ref{sec:app:expt-details} gives further details about them.  

\paragraph{Evaluation protocol and metrics} 
For each of these datasets, we first mark a portion ($\sprot$) of all connections to be protected.
Following previous work~\cite{BackstromL2011SRW, de2013discriminative}, we sort nodes by decreasing number of triangles in which they occur, and allow the top 80\% to be query nodes.
Next, for each query node $q$ in the disclosed graph, the set of nodes $\Vcal\backslash \{q\}$ is partitioned
into  neighbors $\Ncal(q)$ and non-neighbors ${\nnbr(q)}$. 
Finally, we sample 80\% of $|\Ncal(q)|$ neighbors and  80\% of ${\nnbr(q)}$ non-neighbors and
present the resulting graph $\Gcal_{\text{sampled}}$  to $\Acal$---the non-private base LP, and subsequently to $\App$---the \eprivate\ LP which is built on $\Acal$.
%
For each query node $q$, we ask $\App$ to provide a top-$K$ list
of potential neighbors
from the held-out graph, consisting of both public and private connections.
We report the predictive performance of $\App$ in terms of average AUC of the ranked list of nodes across all the queries $Q$. 
The exact choices of $\sprot$ and $K$ vary across different experiments.

\paragraph{Candidates for base LP algorithms~$\Acal$}
We consider two classes of base LP algorithms $\Acal$: 
(i)~algorithms based on the triad-completion principle~\cite{LibenNowellK2007LinkPred} viz., Adamic Adar (AA) and Common Neighbors (CN); and, 
(ii)~algorithms based on fitting node embeddings, viz., GCN~\cite{kipf2016semi} and Node2Vec~\cite{grover2016node2vec}. In Appendix~\ref{sec:app:expt-details}, we 
provide the details about implementation of these algorithms.
Moreover, in Appendix~\ref{sec:app:add}, we also present results on a wide variety of LP protocols.

\paragraph{\our\ and baselines ($\App$)}
We compare \our\ against three state-of-the-art perturbation methods \ie,
Staircase~\cite{geng2015optimal},
Laplace~\citep{MachKDS2011AccuratePrivate}
and Exponential~\citep{MachKDS2011AccuratePrivate,McSherryT2007Mechanism},
which maintain differential privacy off-the-shelf.
Moreover, as an ablation, apart from using $\nu_{\betab}(.)$ --- the linear aggregator of different powers of the score --- as
a signal in $\tf_{\thetab}(.)$ (Eqs.~\ref{eq:dplp-linear} and ~\ref{eq:umnn}), 
we also use it as an independent baseline \ie,  $\tf_{\thetab}(s)=\nu_{\betab}(s)$ with $\thetab=\{\betab\}$.  We refer such a linear ``monotone transformation'' as \our-Lin. In Appendix~\ref{sec:app:add}, as another ablation, we consider another variant of \our, which directly uses the scores $s_{\Acal}(u,v)$ as inputs to UMNN (Eq.~\ref{eq:umnn}), instead of its linear powers $\nu_{\betab}(s_{\Acal}(u,v))$.
Appendix~\ref{sec:app:expt-details} provides the implementation details about \our\ and baselines. 
%
%
\begin{figure}[t]
\centering
 \includegraphics[width=0.3\textwidth]{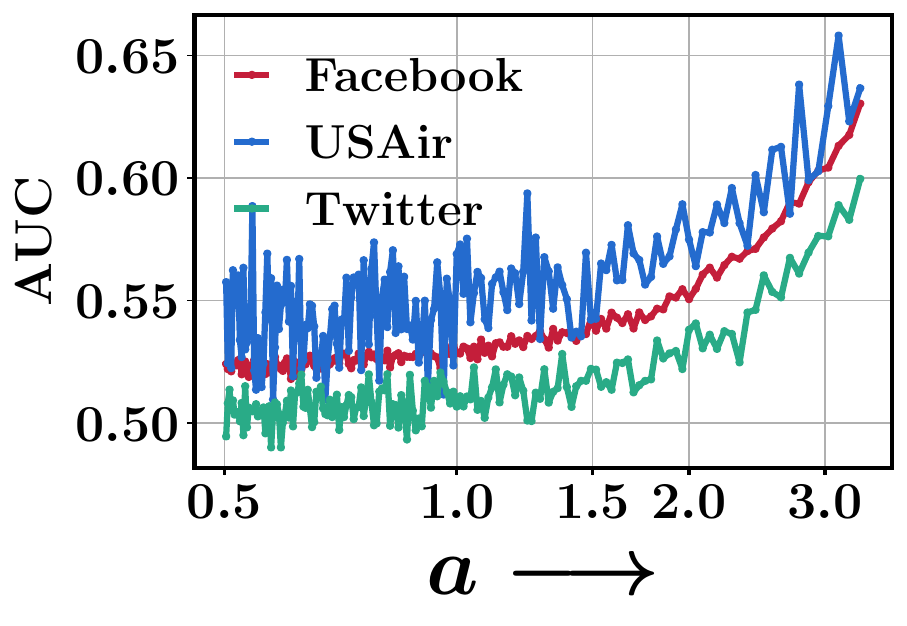}
 \caption{  AUC vs. $a$. $f(s_{\Acal}(u,v))$ $= (s_{\Acal}(u,v))^a$ and the base LP algorithm is \emph{common neighbors} \ie, $\Acal = \tCN$.}  \label{fig:auc-cn}
 \end{figure}
\subsection{Results} 

We first present some empirical evidence which motivates our formulation of $\tf$, specifically, the choice of the 
component $\nu_{\betab}(s_{\Acal}(\cdot,\cdot))$.  We draw nodes using exponential mechanism with different powers of the scores. To recommend neighbors to node $u$, we draw nodes from the distribution in Eq.~\eqref{eq:exp-mech-iterative},
with $\tf(s_{\Acal}(u,\cdot))=(s_{\Acal}(u,\cdot))^{a}$ for different values of~$a$.
Figure~\ref{fig:auc-cn} illustrates the results for three datasets with the fraction of 
protected connections $\sprot=0.3$, the privacy leakage $\epsilon_p=0.1$ and $\Acal=\tCN$.
It shows that AUC improves as $a$ increases.
%
%
This is because, as $a$ increases, 
the underlying sampling distribution induced by the exponential mechanism with new score $\tf(s_\tCN)=s_\tCN ^a$
becomes more and more skewed towards the ranking induced by the base scores $s_{\Acal}$, while maintaining privacy level~$\epsilon_p$.  It turns out that use of UMNN in Eq.~\eqref{eq:umnn} further significantly boosts predictive performance of \our.

We compare \our{} against three
state-of-the-art DP protocols (Staircase, Laplace and exponential), at a given privacy leakage $\epsilon_p=0.1$ and a given fraction of protected edges~$\sprot=0.3$, for top-$K$ (=30) predictions.
\tablename~\ref{tab:FixEpsFindMap} summarizes the results,
which shows that \our{} outperforms competing DP protocols, including its linear variant \our-Lin, for $\Acal=\tAA, \tCN$ and Node2Vec.
There is no consistent winner in case of GCN, which is probably because GCN turns out be not-so-good link predictor in our datasets.

Next, we study how AUC of all algorithms vary as the fraction of protected pairs $\sprot$ and the permissible privacy leakage level $\epsilon_p$ change.  Figure~\ref{fig:var} summarizes the results for Yeast dataset, which shows that our method outperforms the baselines  for almost all values of $\sprot$ and $\epsilon_p$.
\begin{figure}[ht]
  \hspace{-2mm}   \centering
	\subfloat{\includegraphics[width=0.30\textwidth]{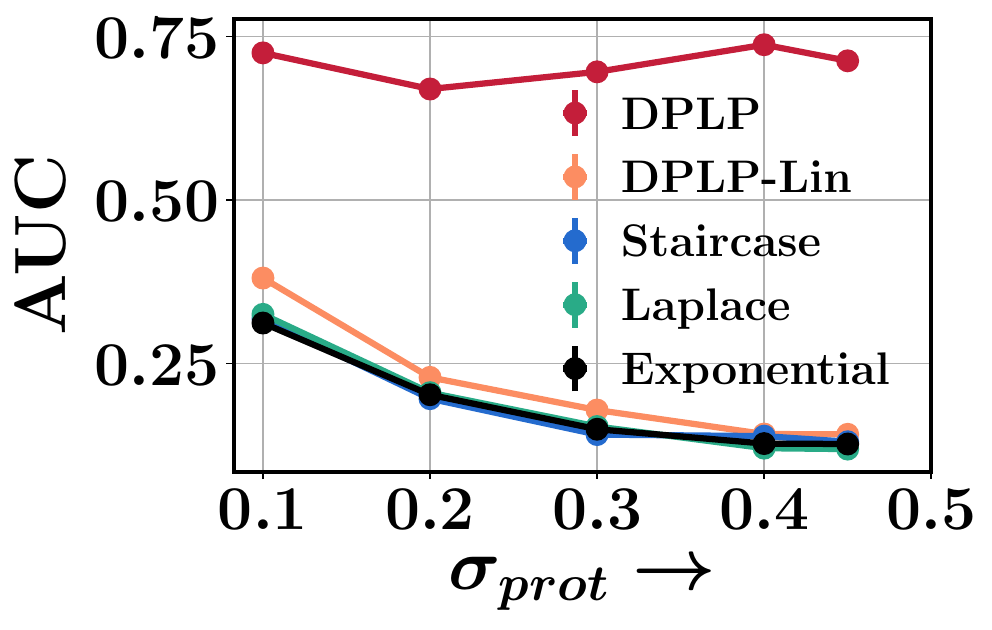}} \hspace{1mm}
	\subfloat{\includegraphics[width=0.28\textwidth]{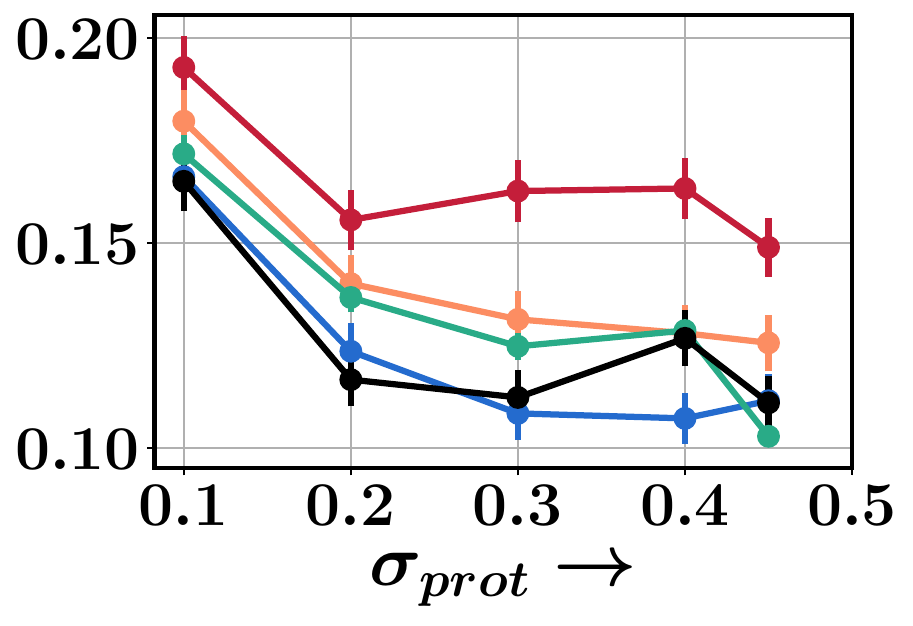}}   \\[-0.1ex]
  \hspace{-4mm}	\centering \subfloat[$\Acal=\tAA$]{\setcounter{subfigure}{1} \includegraphics[width=0.30\textwidth]{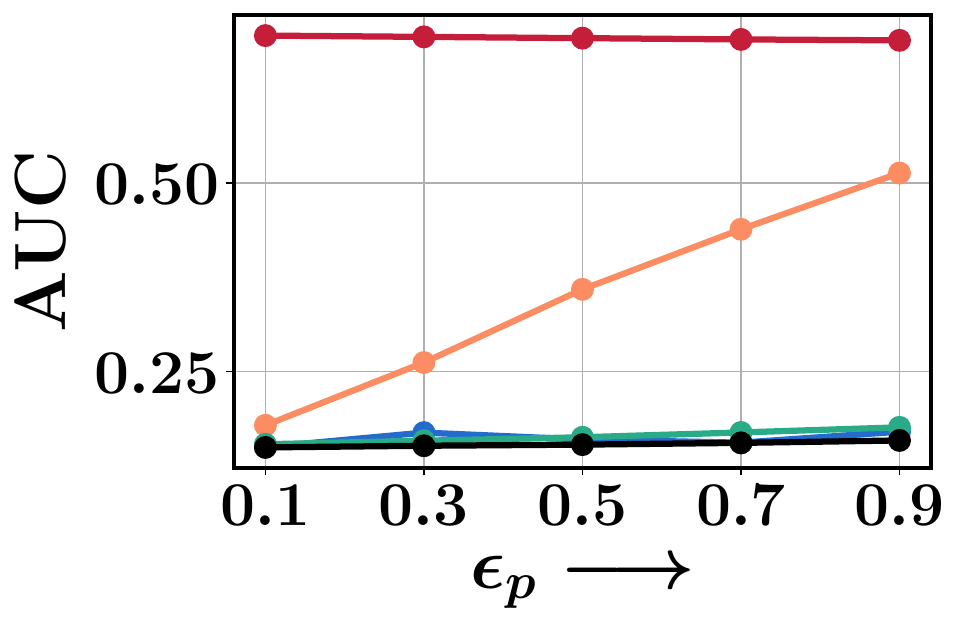}}\hspace{1mm}
	\subfloat[$\Acal=\text{GCN}$]{\includegraphics[width=0.28\textwidth]{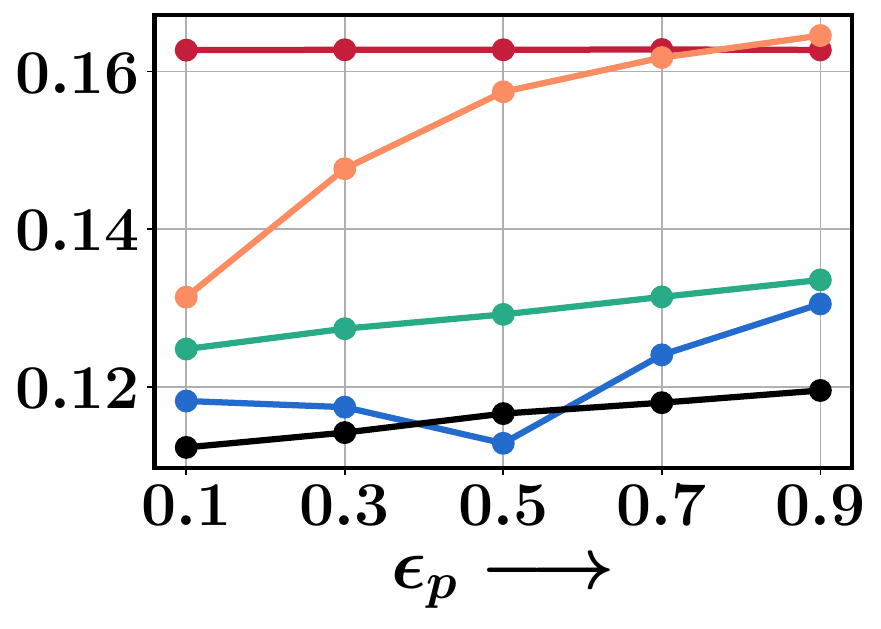}}
\caption{Variation of AUC with $\sprot$ (top row) and with $\epsilon_p$ (bottom row) for all candidates of $\App$
with $\Acal=\tAA$ and $\Acal=\text{GCN}$ for Yeast dataset. For top row, we set $\epsilon_p=0.1$ and for bottom row, we set $\sprot=0.3$.  
} 
\label{fig:var}
\end{figure}

The AUC achieved by \our\ is strikingly stable across a wide variation these parameters. This apparently contradicts the natural notion of utility-privacy tradeoff: increasing utility should decrease privacy.  The apparent conundrum is resolved when we recognize that $\tf_{\thetab}$ adapts to each level of permissible privacy leakage, providing different optimal $\tf$s.  In contrast, Laplace and Exponential mechanisms use the same score values for different $\epsilon_p$.  Therefore, as $\epsilon_p$ decreases, the underlying perturbation mechanism moves quickly toward the uniform distribution.
While both Staircase and \our-Lin aim to find optimal noise distributions, Staircase is not optimized under ranking utility and \our-Lin is constrained by its limited expressiveness. Therefore, they fare poorly.

\section{Conclusion}
\label{sec:End}

We proposed \eprivacy, a practically motivated variant of differential privacy for social networks.  Then we presented \our, a noising protocol designed around a base non-private LP algorithm. \our{} maximizes ranking accuracy, while maintaining a specified privacy level of protected pairs. It transforms the base score using a trainable monotone neural network and then introduces noise into these transformed scores to perturb the ranking distribution.  \our\ is trained using a pairwise ranking loss function.
We also analyzed the loss of ranking quality in a latent distance graph generative framework. Extensive experiments show that \our\ trades off ranking accuracy and privacy better than several baselines.
Our work opens up several interesting directions for future work. 
E.g., one may investigate collusion between queries and graph
steganography attacks~\citep{BackstromDK2007GraphSteganography}.
One may also analyze utility-privacy trade-off in
other graph models \eg, Barab{\'a}si-Albert graphs~\cite{barabasi1999emergence}, Kronecker graphs~\cite{leskovec2010kronecker}, etc.

\section{Broader Impact}
\label{sec:Broader}
Privacy concerns
have increased with the proliferation of social media. 
Today, law enforcement authorities routinely demand
social media handles of visa applicants~\cite{impacta}.  
Insurance providers can use social media content to set premium levels~\cite{impactb}.
One may avoid such situations, by making individual features \eg, 
demographic information, age, etc. private.
However, homophily and other
network effects may still leak user attributes, even if they are made explicitly private.
Someone having skydiver or smoker friends may end up paying large insurance
premiums, even if they do not subject themselves to those risks. A potential remedy is for users to mark some part of
their connections as private as well~\citep{waniek2019hide}.
In our work, we argue that it is possible to protect privacy of such protected connections,
for a specific graph based application, \ie, link prediction, while not sacrificing the quality of prediction.
More broadly, our work suggests that, a suitably designed \dpp\ algorithm can still meet the user expectation
in social network applications, without leaking user information.

\section*{Acknowledgement} Both the authors would like to acknowledge IBM Grants. Abir De would like to acknowledge Seed Grant provided by IIT Bombay.

\setlength{\bibsep}{0pt plus 0.1ex}
\bibliographystyle{abbrvnat}
\bibliography{refs,voila}

\newpage

\begin{center}
\bfseries \Large \papertitle \\[.5ex]
\Large (Appendix)
\normalsize 
\end{center}

\section*{Roadmap to different parts of appendix}
\begin{itemize}
    \item Appendix~\ref{app:backgroundAppen} contains formal definitions of node differential privacy, edge differential privacy and \eprivacy\ in a broader context.
    \item Appendix~\ref{sec:connectionAppen} connects  protected-node-pair differential privacy with other variants of differential privacy.
    \item Appendix~\ref{sec:appen:dpone:proof} contains proof of Proposition~\ref{thm:dpone}, \ie, the privacy guarantee of Algorithm~\ref{alg:dpone}.
    \item Appendix~\ref{sec:main-thm-proof} shows the proofs of the technical results presented in Section~\ref{sec:Qual}. More specifically, here we bound the relative (Proposition~\ref{prop:relative-tradeoff}) and absolute (Theorem~\ref{thm:util-main-thm-dp}) trade-off between privacy and LP accuracy.
    \item Appendix~\ref{sec:AuxLemmas} provides auxiliary theoretical results, required to prove the results of Appendix~\ref{sec:main-thm-proof}.
    \item Appendix~\ref{sec:app:expt-details} provides additional experimental details. 
    \item Appendix~\ref{sec:app:add} provides additional experiments.
\end{itemize}
\vspace{1cm}
\begin{figure}[h!]
    \centering
    \includegraphics[width=0.7\textwidth]{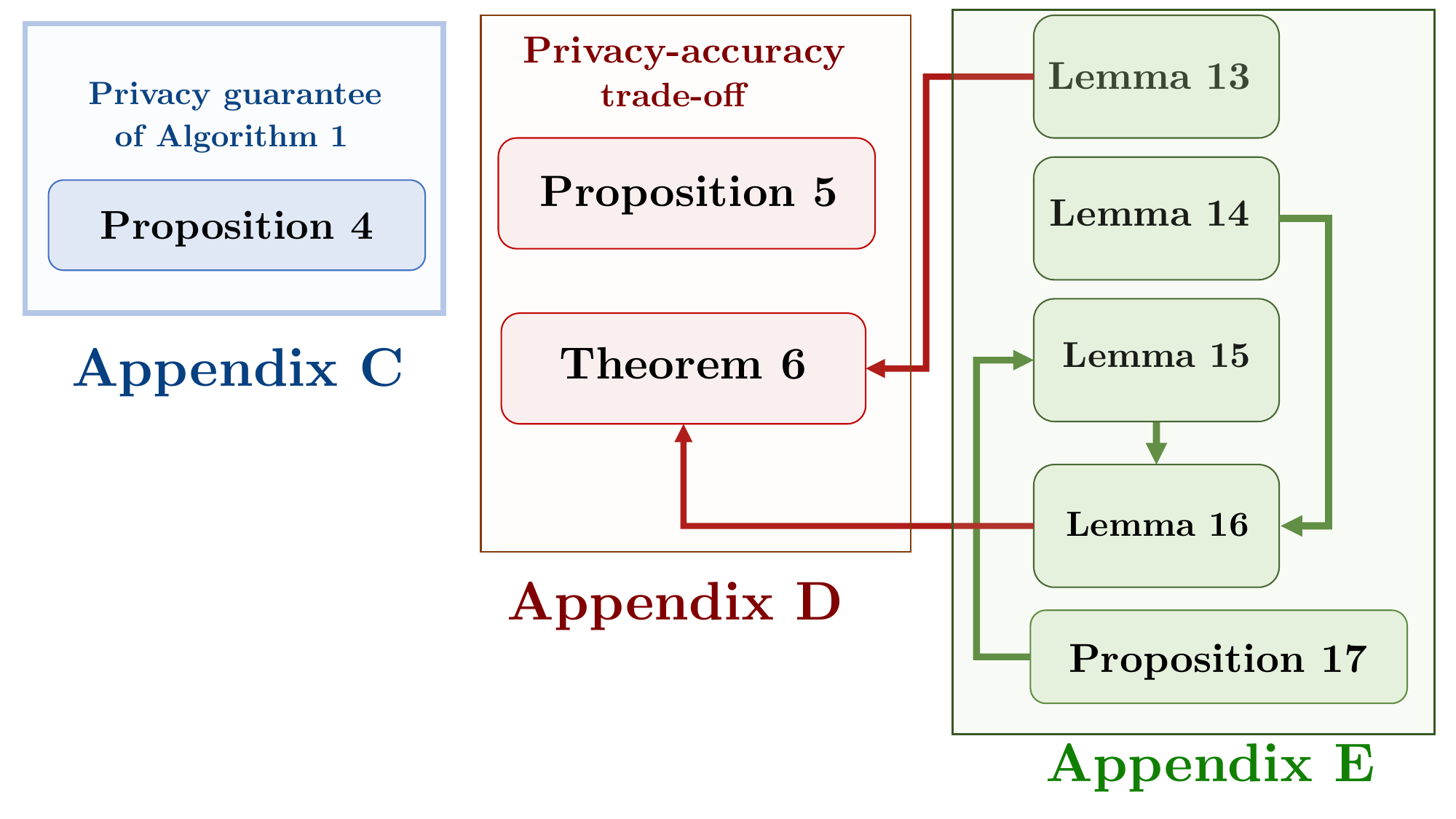}
    \caption{Inter-dependence between the theoretical results presented in Appendix. Appendix C proves privacy guarantee of Algorithm~\ref{alg:dpone}.
    Appendix D contains proofs on privacy-utility trade-off. Appendix E contains auxiliary results. \boxed{A}$\to$\boxed{B} means that $A$ is used to prove $B$.}
    \label{fig:theoretical-dag}
\end{figure}
\newpage
\section{Formal definition of node differential privacy, edge differential privacy,\\
\eprivacy}\label{app:backgroundAppen}
Any algorithm $\Mcal$ on a graph is $\epsilon_p$ differentially private~\citep{MachKDS2011AccuratePrivate},
if, we have
${\Pr(\Mcal(\Gcal) \in \Ocal)} \le \exp(\epsilon_p |\Gcal \oplus \Gcal'| ) {\Pr(\Mcal(\Gcal') \in \Ocal)} $
where $\Gcal'$ (called as a neighbor graph of $\Gcal$)
is obtained by perturbing $\Gcal$.   $|\Gcal \oplus \Gcal'|$
measures the extent of perturbation between $\Gcal$ and $\Gcal'$.
In order to apply DP on graphs, we need to
quantify the granularity of such a perturbation, or in other words, define
the neighborhood criteria of two graphs $\Gcal$ and $\Gcal'$.
Existing work predominantly considers two types of DP formulations on graphs: (i)~edge differential privacy (edge DP), where $\Gcal$ and $\Gcal'$ differ by one edge~\citep{nissim2007smooth,hay2009accurate} and, (ii)~node differential privacy (node DP), where $\Gcal$ and $\Gcal'$ differ by a node and all its incident edges~\citep{borgs2015private,kasiviswanathan2013analyzing,song2018differentially}.
%
%

\begin{definition}[\bfseries Node DP]
Two graphs $\mathcal{G}=(\Vcal,\Ecal)$ and $\mathcal{G}'=(\Vcal,\Ecal')$ with the same node set $\Vcal$ are \emph{neighboring} if $\Ecal\ne\Ecal'$ and there exists a node $u$ such that $\Ecal\cp \{(u,v)|v\in\Vcal\cp\{u\}\}=\Ecal'\cp \{(u,v)|v\in\Vcal\cp\{u\}\}$. Then an algorithm $\Mcal$ is $\epsilon_p$ node DP if, for any two neighboring graphs $\Gcal,\Gcal'$, we have:
\begin{align}
    \Pr(\Mcal(\Gcal)\in \mathcal{O})\le e^{\epsilon_p} \Pr(\Mcal(\Gcal')\in \mathcal{O})
\end{align}
\end{definition}

\begin{definition}[\bfseries  Edge DP]
Two graphs $\mathcal{G}=(\Vcal,\Ecal)$ and $\mathcal{G}'=(\Vcal,\Ecal')$ with the same node set $\Vcal$ are \emph{neighboring} if $\Ecal\ne\Ecal'$ and there exists a node pair $(u,v)$, such that $\Ecal\cp \{(u,v)\}=\Ecal'\cp \{(u,v)\}$. Then an algorithm $\Mcal$ is $\epsilon_p$ edge DP if, for any two neighboring graphs $\Gcal,\Gcal'$ we have:
\begin{align}
    \Pr(\Mcal(\Gcal)\in \mathcal{O})\le e^{\epsilon_p} \Pr(\Mcal(\Gcal')\in \mathcal{O})
\end{align}
\end{definition}

\begin{definition}[\bfseries Protected-node-pair DP] \label{def:epp-appen}
We are given  a randomized algorithm $\Mcal$ which operates on each node $u\in\Vcal$.  We say that ``$\Mcal$ ensures $\epsilon_p$-\eprivacy'', if, for all nodes $u\in\Vcal$ we have 
\begin{align}
\Pr(\Mcal(\Gcal,u)\in \Ocal ) \le e^{\epsilon_p} \Pr(\Mcal(\Gcal',u)\in \Ocal),
\end{align}
whenever $\Gcal$ and $\Gcal'$ are $u$-neighboring graphs (See Definition~\ref{def:uNbrGraphs}).  
\end{definition}

Neighboring graphs $\Gcal,\Gcal'$ have $|\Gcal \oplus \Gcal'|=1$ for node~DP, edge~DP and Protected-node-pair DP. \todo{@ad check}



\section{Connection of  protected-node-pair differential privacy with other variants of differential privacy}\label{sec:connectionAppen}
In the following, we relate our proposed privacy model on graphs with three variants of differential privacy \eg, group privacy, Pufferfish privacy and Blowfish privacy.

\paragraph{Connection with group privacy}
Our privacy proposal can be cast as an instance of group privacy~\cite{dwork2011differential,song2017pufferfish}, where the goal is to conceal the information of protected groups.

\begin{definition}[\bfseries Group privacy]
Given a database $\Dcal$ and a collection of subsets containing sensitive information: 
$\textsc{Priv} = \{P_1, P_2,\cdots,P_k \}$
with  $P_i\subset \Dcal$. 
An algorithm $\Mcal$ is $\epsilon_p$ group \dpp\ w.r.t $\textsc{Priv}$ if for all $P_i$ in $\textsc{Priv}$ for all pairs of $\Dcal$ and $\Dcal'$ which differ in the entries in $P_i$ \ie, $\Dcal\cp P_i = \Dcal'\cp P_i$, we have:
\begin{align}
    \Pr(\Mcal(\mathcal{D})\in \mathcal{O})\le e^{\epsilon_p} \Pr(\Mcal(\mathcal{D}')\in \mathcal{O}) \ \forall \ \mathcal{O}\in \text{Range}(\Mcal) 
\end{align}
\end{definition}
Note that in the context of our current problem, if we have $\Dcal=\Gcal, \Dcal'=\Gcal'$, with $G$ and $G'$ being $u$-neighboring graphs; $P_w=\set{(v,w) | v \in \Pcal(w)\cp u}$; and, $\Mcal$ to be a randomized LP algorithm   which operates on $u$, then, Definition~\ref{def:epp-appen} shows that our privacy model closely related group privacy. However, note that since the sensitive information differs across different nodes, the set $\textsc{Priv}$ also varies across different nodes.

\paragraph{Connection with Pufferfish privacy}
Our privacy proposal is also connected with Pufferfish privacy~\cite{kifer2014pufferfish,song2017pufferfish}. In Pufferfish privacy, we define three parameters: $\mathcal{S}$ which contains the sensitive information; $\Qcal \subseteq \Scal\times\Scal$ which contains the set of pairs, that we want to be indistinguishable; and, $\Gamma$ a set data distributions. In terms of these parameters, we define Pufferfish privacy as follows.

\begin{definition}[\bfseries Pufferfish privacy~\citep{kifer2014pufferfish,song2017pufferfish}]
Given the tuple \todo{where is $\Qcal$ used?} $(\Scal, \Qcal, \Gamma)$, an algorithm $\Mcal$ is said to be $\epsilon_p$ Pufferfish private if for all $\zeta\in\Gamma$ with $\Dcal$ drawn from distribution $\zeta$, for all $(a_i,a_j)\in \Qcal$, we have:
\begin{align}
\log \left|    \dfrac{\PP(\Mcal(\Dcal)\in\Ocal \,|\, a_i,\zeta)}{\PP(\Mcal(\Dcal) \in\Ocal \,|\, a_j,\zeta)} \right| < \epsilon_p  \ \forall\ \Ocal\ \in \text{Range}(\Mcal)\ \text{ and } \PP(a_i|\zeta)\neq 0, \ \PP(a_j|\zeta)\neq 0.
\end{align}
\end{definition}
In the context of social networks, where users have marked private edges, one can naturally model $\Scal$ to be the set of protected node-pairs \ie, $\Scal=\left\{(u,u')\,|\, u'\in\Pcal(u)\right\}$ and $\Qcal=\left\{\left(\indicator{u'\in \Ncal(u)}, \indicator{u'\in \nbu} \right) \right\}$ and $\gamma$ can incorporate the adversary's belief about the generative model of the graph. 
However, we do not incorporate the data distribution $\zeta$ but instead wish to make our DP algorithm indistinguishable w.r.t.\ the entire possible set $\Pcal(u)$, not just one single instance of presence or absence of an edge. It would be interesting to consider a Pufferfish privacy framework for LP algorithms in graphs, which is left for future work.

\paragraph{Connection with Blowfish privacy}
Blowfish privacy~\cite{he2014blowfish} is another variant of differential privacy, which is motivated by Pufferfish privacy. However, in addition to modeling sensitive information, it characterizes the constraints which may be publicly known about the data. Moreover, it does not directly model any distribution about the data. In contrast to our work, the work in~\cite{he2014blowfish} neither considers a ranking problem in a graph setting nor aims to provide any algorithm to optimize privacy-utility trade off in a graph database.

\section{Privacy guarantee of Algorithm~\ref{alg:dpone}}
\label{sec:appen:dponeX}
%
\label{sec:appen:dpone:proof}
\textbf{Proposition {\ref{thm:dpone}}.}
\emph{$\App$ is $K\epsilon_p$-protected-node-pair differentially private.}
\begin{proof}
 First note that line no. 1--5 in Algorithm~\ref{alg:dpone} do not have any privacy leakage. 
 For the recommendation on the candidate nodes,the derivation of the privacy guarantee is identical to~\cite[Page 43]{DworkR2014AlgoDP}.  However, it is reproduced here for completeness. Let $u_j=\Rcal_K(j)$ for the graph $\Gcal$ and $u'_j=\Rcal_K(j)$ for the graph. Moreover, we denote that $\Rcal_K(1...j-1)=[u_1,...,u_{j-1}]$ for $\Gcal$ and $\Rcal_K(1...j-1)=[u' _1,...,u' _{j-1}]$ for $\Gcal'$. 
\begin{align}
    \dfrac{\Pr(\pi_u ^\App =\Rcal_K |\Acal,\Gcal)}{\Pr(\pi_u ^\App =\Rcal_K |\Acal,\Gcal')}& =    \dfrac{\displaystyle{\prod_{k=1} ^K} \dfrac{\exp\left(\dfrac{\epsilon_p \tf(s_\Acal(u, u_{k}|\Gcal))}{2\Deltaf}\right)}
       {\displaystyle{\sum_{w \not \in  \Rcal_K(1,...,j-1)}}
         \exp\left(\dfrac{\epsilon_p \tf(s_\Acal(u,w) |\Gcal )}{2\Deltaf}\right)}}{\displaystyle{\prod_{j=1} ^K} \dfrac{\exp\left(\dfrac{\epsilon_p \tf(s_\Acal(u, u' _{j})|\Gcal')}{2\Deltaf}\right)}
       {\displaystyle{  \sum_{w' \not \in \Rcal_K(1,..,j-1)}}
         \exp\left(\dfrac{\epsilon_p \tf(s_\Acal(u,w')|\Gcal' )}{2\Deltaf}\right)}}\nn\\
         & = \prod_{j=1} ^K  \exp\left(\dfrac{\epsilon_p \big(\tf(s_\Acal(u, u_{k}|\Gcal) - \tf(s_\Acal(u,u'_k)|\Gcal')\big)}{2\Deltaf}\right) \times 
          \displaystyle{\prod_{j=1} ^K} \dfrac{ {\displaystyle{\sum_{w' \not\in \Rcal_K(1,..,j-1)}}\exp\left(\dfrac{\epsilon_p  \tf(s_\Acal(u, w'|\Gcal)}{2\Deltaf}\right)} }{\displaystyle{\sum_{w \not \in \Rcal_K(1,..,j-1)}}\exp\left(\dfrac{\epsilon_p  \tf(s_\Acal(u, w|\Gcal))}{2\Deltaf}\right)}\nn\\
         & \le \exp(K\epsilon_p/2). \exp(K\epsilon_p/2) = \exp(K\epsilon_p)
\end{align}
The last inequality follows from the fact that:
$\tf(s_\Acal(u, a)|\Gcal) -\tf(s_\Acal(u, a)|\Gcal') \le \Delta_{f,\Acal}$.
\end{proof}
\section{Proof of technical results in Section~\ref{sec:Qual}}
\label{sec:main-thm-proof}

\subsection{Relative error in score due to Algorithm~\ref{alg:dpone}}
\label{sec:gamma-basic-proof}

\textbf{Proposition 5.}
If $\epsilon_p$ is the privacy parameter and $\kappa_{u,\Acal}:=\max_{i\in [|\Vcal|-1]} \set{s_{\Acal}(u,\uaa_i)- s_{\Acal}(u,\uaa_{i+1})}$, then we have:
   \begin{align}
       &\lim_{\epsilon_p\to \infty}\util = 0  \text{ and, }
        \lim_{\epsilon_p\to 0 }\util \le   K \kappa_{u,\Acal}. \nn
   \end{align}
%
\begin{proof}
Recall that: $\util$ was defined as:
\begin{align}
  \util= \EE_{\App}\Big[\sum_{i\in[K]} (s_{\Acal} (u,\uaa_i|\Gcal) -s_{\Acal} (u,\upv _i|\Gcal))\Big]
\end{align}
We note that:
\begin{align}
& \EE_{\App}\left(s_{\Acal}(u,\uaa_i)-s_{\Acal}(u,u^{\App} _i )\right) \\
 & =\sum_{\Rcal_K{(1,...,i-1)}} \sum_{w\not \in \Rcal_K{(1,...,i-1)}} \left(s_{\Acal}(u,\uaa_i)-s_{\Acal}(u,w )\right) \frac{\exp\left(\dfrac{\epsilon_p \tf(s_{\Acal}(u,w))}{2\Deltaf} \right) }{ \displaystyle{\sum_{v\not\in \Rcal_K{(1,...,i-1)}}}\exp\left(\dfrac{\epsilon_p \tf(s_{\Acal}(u,v))}{2\Deltaf} \right)  }\PP(\Rcal_K{(1,...,i-1)})
 \end{align}
(1) \boxed{\lim_{\epsilon_p\to \infty}\util = 0.} 
 We use induction to prove that $  \lim_{\epsilon_p\to \infty}
  \EE_{\App}\left(s_{\Acal}(u,\uaa_i)-s_{\Acal}(u,u^{\App} _i )\right)=0$ for all $i$. For $i=1$, we have:
\begin{align}
  \lim_{\epsilon_p\to \infty}
  \EE_{\App}\left(s_{\Acal}(u,\uaa_1)-s_{\Acal}(u,u^{\App} _1 )\right) & =  \sum_{w\in\Vcal} \left(s_{\Acal}(u,\uaa_1)-s_{\Acal}(u,w)\right) \mathbf{1}(w=u_1 ^\Acal)  
   =0
\end{align}
Suppose we have $  \lim_{\epsilon_p\to \infty}
  \EE_{\App}\left(s_{\Acal}(u,\uaa_j)-s_{\Acal}(u,u^{\App} _j )\right)=0$ for $j=1,...,i-1$. To prove this for $j=i$, we note that $\Rcal_K(j)= u_j ^\Acal$ for $j\in[i-1]$. Hence, we have:
\begin{align}
 & \lim_{\epsilon_p\to \infty}  \EE_{\App}\left(s_{\Acal}(u,\uaa_i)-s_{\Acal}(u,u^{\App} _i )\right)\nn \\
 & =   \sum_{\Rcal_K{(1,...,i-1)}}
 \sum_{w\not \in \Rcal_K{(1,...,i-1)}}
 \left(s_{\Acal}(u,\uaa_i)-s_{\Acal}(u,w )\right) 
 \lim_{\epsilon_p\to \infty}
 \frac{\exp\left(\dfrac{\epsilon_p \tf(s_{\Acal}(u,w))}{2\Deltaf} \right) }
 { \displaystyle{\sum_{v\not\in \Rcal_K{(1,...,i-1)}}}\exp\left(\dfrac{\epsilon_p \tf(s_{\Acal}(u,v))}{2\Deltaf} \right)  }\PP(\Rcal_K{(1,...,i-1)})\nn\\
& =  \sum_{w\not \in \set{u_j ^\Acal | j\in[i-1]}} \left(s_{\Acal}(u,\uaa_i)-s_{\Acal}(u,w )\right) \lim_{\epsilon_p\to \infty} \frac{\exp\left(\dfrac{\epsilon_p \tf(s_{\Acal}(u,w))}{2\Deltaf} \right) }{ \displaystyle{ \sum_{v\not \in \set{u_i ^\Acal | j\in[i-1]}}}\exp\left(\dfrac{\epsilon_p \tf(s_{\Acal}(u,v))}{2\Deltaf} \right)  } \nn\\
& =  \sum_{w\not \in \set{u_j ^\Acal | j\in[i-1]}} \left(s_{\Acal}(u,\uaa_i)-s_{\Acal}(u,w )\right) \mathbf{1}(w=u_i ^\Acal) \nn\\
  &=0
\end{align}
(2) \boxed{ \lim_{\epsilon_p\to 0 }\util \le   K \kappa_{u,\Acal}.} 
Making $\epsilon_p\to 0$ we have 
\begin{align}
 \lim_{\epsilon_p\to 0}    \EE_{\App}\left(s_{\Acal}(u,\uaa_i)-s_{\Acal}(u,u^{\App} _i )\right)
& \le \sum_{\Rcal_K{(1,...,i-1)}} \sum_{w\not \in \Rcal_K{(1,...,i-1)}} \dfrac{\kappa_{u,\Acal}}{|\Vcal|-i+1}
\PP(\Rcal_K{(1,...,i-1)})\nn\\
& \sum_{\Rcal_K{(1,...,i-1)}} (|\Vcal|-i+1)\dfrac{\kappa_{u,\Acal}}{|\Vcal|-i+1}
\PP(\Rcal_K{(1,...,i-1)})\nn\\
& = \kappa_{u,\Acal}
\end{align}
Hence, 
$\lim_{\epsilon_p\to 0} \util=\sum_{i\in[K]} \EE_{\App}\left(s_{\Acal}(u,\uaa_i)-s_{\Acal}(u,u^{\App} _i )\right) \le
K\kappa_{u,\Acal}$.
\end{proof}

\subsection{Proof of Theorem~\ref{thm:util-main-thm-dp}}
\begin{numtheorem}{\ref{thm:util-main-thm-dp}}
   Given $\epsilon=\sqrt{\frac {2\log (2/\delta)}{|\Vcal|}}+\frac{7 \log (2/\delta)}{3(|\Vcal|-1)}$. For $\Acal\in\set{\tAA,\tCN,\tJC}$, with probability $1-4K^2\delta$:
  \begin{align}
      \EE_{{\App}}\left[\rlp\right]= O\left(\big[2K\epsilon+ {\util }/{|\Vcal|}\big]^{\frac{1}{KD}} \right),\nn
  \end{align}
  \end{numtheorem}
\xhdr{Proof sketch} To prove the theorem, we follow the following steps. First we use Lemma~\ref{lem:lastX} to bound $(d_{u u^\App _t}-  d_{uu^* _t})$ for each $t\in[K]$, where $u^*_t = \pi^* _u(t)$ is the $t$-th node in the ranked list given by the oracle ranking $\pi^* _u$ which sorts the node in the increasing order of the distances $d_{u\bullet}$. Next, we use the bounds obtained in above step and Lemma~\ref{lem:ranking-loss-bound-1} to bound $ \EE_{\App}\left(\rlp\right)$.
\begin{proof}
Here, we show only the case of common neighbors.  Others follow using
the same method.  
We first fix one realization of $\App$ for $\Acal=\tCN$.\\

\noindent First, we define the following quantities: 
\begin{align}
    & \utildpA =\sum_{i\in[K]}\left(s_{\Acal}(u,\uaa_i)-s_{\Acal}(u,\upv_i)\right) \quad \bigg(\text{Hence, }\EE_{\App}(\utildpA) = \util \bigg)\\
    & {\epsilon}_{\App}=\epsilon+ \dfrac{\utildpCN}{2K|\Vcal|}, \\ &\epsilon_{\App,t}=2rt\left(\frac{2t\epsilon_{\App}}{\Omega(r)}\right)^{\frac{1}{tD}} \label{eq:volume-intro}\\
    & S_t=\sum_{i=1} ^t (d_{uu^\App _i}-d_{uu^* _i}). \quad \bigg(\text{Hence, } S_t-S_{t-1} = ( d_{uu^\App _t} -  d_{uu^* _t} )\bigg),\label{eq:delst}
\end{align}
where in Eq.~\eqref{eq:volume-intro}, $\Omega(r)$ indicates the volume of a $D$ dimensional hypersphere with radius $r$;
in Eq.~\eqref{eq:delst}, $u^*_i = \pi^* _u(i)$ is the $i$-th node in the ranked list given by the oracle ranking $\pi^* _u$ which sorts the node in the increasing order of the distances $d_{u\bullet}$.

\noindent \boxed{\text{Bounding } S_t-S_{t-1} = (  d_{u u^\App _t}  - d_{uu^* _t}  )    .}
We note that:
\begin{align}
\PP(|S_t- & S_{t-1}| < \epsilon_{\App,t}) \ge 1-\PP(S_t \ge \epsilon_{\App,t})-\PP(S_{t-1} \ge \epsilon_{\App,t-1}) \quad \big(\text{Due to: } S_t\le \epsilon_{\App,t}\ \forall t \implies -\epsilon_{\App, t-1}<S_t-S_{t-1}<\epsilon_{\App,t}\big) \nn \\
&\ge 1-4t\delta \qquad \text{(Due to Lemma~\ref{lem:lastX})}\label{eq:intermed-1}. 
\end{align}
The above equation shows that:
\begin{align}
&\PP\bigg(\sum_{t\in[K]}|S_t-S_{t-1}|^2 <  K \epsilon^2 _{\App,K}\bigg) 
 \overset{(i)}{\ge} \Pr \bigg(\sum_{t\in[K]}|S_t-S_{t-1}|^2< \sum_{t\in[K]} \epsilon^2 _{\App,t}\bigg) \nn\\
&\hspace{5.2cm} {\ge}  1-\sum_{t\in[K]} \PP \bigg( |S_t-S_{t-1}|^2 < \epsilon^2 _{\App,t}\bigg) \quad \nn\\
&\hspace{5.2cm}{\ge} 1 -2K(K+1)\delta \qquad (\text{Due to Eq.~\eqref{eq:intermed-1}}) \nn\\
&\hspace{5.2cm}{\ge} 1-4K^2\delta. \label{eq:intermed-2}
\end{align}
Inequality (i) is obtained using
$\displaystyle    \epsilon_{\App,t} = 2rt\left(\frac{2t\epsilon_{\App}}{\Omega(r)}\right)^{\frac{1}{tD}} 
   \overset{(a)}{\le }\ 2Kr\left(\frac{2K\epsilon_{\App}}{\Omega(r)}\right)^{\frac{1}{tD}} 
   \overset{(b)}{\le} \ 2Kr\left(\frac{2K\epsilon_{\App}}{\Omega(r)}\right)^{\frac{1}{KD}} = \epsilon_{\App,K}$, where
the first inequality $(a)$ is due to that $t\le K$ and the second inequality $(b)$ is due to that $K\epsilon_{\App} \ll 1$. 

\noindent \boxed{\text{Bounding }\EE_{\App}\left(\rlp\right).}
From~\eqref{eq:intermed-2}, we have the following with probability $1-4K^2\delta$:
 \begin{align}
    \EE_{\App}\left[\textstyle
     \sum_{t\in[K]}|S_t-S_{t-1}|^2\right]& \le     \EE_{\App}(\epsilon^2 _{\App,K})
     \le  \boundfac   \left(\frac{2K\epsilon+ \EE_{\App}(\utildpCN)/|\Vcal|}{\Omega(r)}\right)^{2/KD},
 \end{align}
where $r$ is the threshold distance for an edge
between two nodes in the latent space model. The last inequality is obtained using Jensen inequality on the concave function $f(x) = x^{1/KD}$. Further  applying Jensen inequality on the bounds of Lemma~\ref{lem:ranking-loss-bound-1} gives us:
\begin{align}
     \EE_{\App}\bigg(\rlp\bigg) \le \sqrt{2 {K \choose 2}}\sqrt{ K  \sum_{i\in[K]} \EE_{\App}\left(d_{uu^{\App} _i} - d_{uu^*_i}\right)^2  },
\end{align}
which immediately proves the result.
\end{proof}

  \section{Auxiliary Lemmas}
\label{sec:AuxLemmas}
In this section, we first provide a set of key auxiliary lemmas that will be used to derive several results in the paper. We first provide few definitions.

\begin{definition}
Given the generative graph model with nodes $\Vcal$, the radius of connectivity $r$. 
\begin{enumerate}
\item We define score deviation $\utildpA=\sum_{i\in[K]}(s_{\Acal}(u,\uaa_i)-s_{\Acal}(u,\upv_i)).$
\item We define $\Omega(r)$ as the volume of a $D$ dimensional hypersphere
of radius $r$.
\item We define $\pi^* _u$ as the Oracle ranking which sorts the node in the increasing order of the distances $d_{u\bullet}$. Moreover, we denote $u^*_i = \pi^* _u(i)$ as the $i$-th node in the ranking.
\item Let $u,v$ be nodes in the graph with corresponding point embeddings.  With these points as centers, two $D$-dimensional hyperspheres, each with radius~$r$, are described.  The common volume of the intersection of the two hyperspheres is called $\Area(u,v) = \Area(v,u)$.
\end{enumerate}
\end{definition}

\begin{lemma}\label{lem:ranking-loss-bound-1}
We have:
\begin{align}
    \rlp\le \sqrt{2 {K \choose 2}}\sqrt{ K  \sum_{i\in[K]} \left(d_{uu^{\App} _i} - d_{uu^*_i}\right)^2  }
\end{align}
\end{lemma}
\begin{proof}
Let the nodes be $u^{\App} _1, \ldots, u^{\App} _K, \ldots$ in the ranked list.
We are not concerned with positions after~$K$. 
Let the latent distances from $u$ to these
nodes be $d_{uu^{\App} _1}, d_{uu^{\App} _2}, \ldots, d_{uu^{\App} _K}$.  If $i<j$ then we want
$d_{uu^{\App} _i} < d_{uu^{\App} _j}$, but this may not happen, in which case we assess a loss
of $(d_{uu^{\App} _i} - d_{uu^{\App} _j})$. Hence, we have:
\begin{align}
\rlp & = \sum_{i < j \le K} \left[d_{uu^{\App} _i}-d_{uu^{\App} _j}\right]_+ \nn\\
    & \overset{(i)}{=} \sum_{i < j \le K} \left[d_{uu^{\App} _i}-d_{uu^{\App} _j}\right]
            \indicator{\pi^*(u^\App _i) > \pi^*(u^\App _i)} \nn\\
    &  \overset{(ii)}{\le} \sqrt{K \choose 2}\sqrt{ \sum_{i < j \le K} \left[d_{uu^{\App} _i}-d_{uu^{\App} _j}\right]^2
            \indicator{\pi^*(u^\App _i) \ge \pi^*(u^\App _i)} },         
\end{align}
where, (i) is due to the fact that 
$\indicator{ d_{uu^{\App} _i} > d_{uu^{\App} _j} }
\iff \indicator{ \pi^*_u(u^{\App} _i) > \pi^* _u(u^{\App} _j) }$ and (ii) is due to Cauchy-Schwartz inequality, \ie,
 $|\sum_i x_i \sum_j y_j| \le \sqrt{\sum_i x^2 _i} \sqrt{\sum_j x^2 _j}$.

Now since $i< j$, we note that $d_{uu^* _j}>d_{uu^* _i}$. Therefore, 
$d_{uu^{\App} _i} - d_{uu^{\App} _j} \le d_{uu^{\App} _i} - d_{uu^{\App} _j}
+d_{uu^*_j} - d_{uu^*_i}$. Now if $\pi^* _u(u^{\App} _i) - \pi^* _u(u^{\App} _j) > 0$, then  $d_{uu^{\App} _i} - d_{uu^{\App} _j}>0$.
Hence, we have:
\[|d_{uu^{\App} _i} - d_{uu^{\App} _j}|=d_{uu^{\App} _i} - d_{uu^{\App} _j} \le d_{uu^{\App} _i} - d_{uu^{\App} _j}
+d_{uu^*_j} - d_{uu^*_i} \le |d_{uu^{\App} _i} - d_{uu^* _i}|
+ |d_{uu^{\App}  _j} - d_{uu^*_j}| . \]
Now we have:
\begin{align}
 & \sum_{i<j\le K} (d_{uu^{\App} _i} - d_{uu^{\App} _j})^2
\indicator{ \pi^* _u(u^{\App} _i) - \pi^* _u(u^{\App} _j) > 0}\nn\\
&\le \sum_{i<j\le K} \left(|d_{uu^{\App} _i} - d_{uu^*_i}|
+ |d_{uu^{\App}  _j} - d_{uu^*_j}|\right)^2\nn\\
&\le 2 \sum_{i<j\le K}\left[ \left(d_{uu^{\App} _i} - d_{uu^*_i}\right)^2
+ \left(d_{uu^{\App}  _j} - d_{uu^*_j}\right)^2 \right]\nn\\
& \le 2K \sum_{i\in[K]} \left(d_{uu^{\App} _i} - d_{uu^*_i}\right)^2
\end{align}
Finally, we have:
\begin{align}
    \rlp& \le \sqrt{K \choose 2}\sqrt{ \sum_{i < j \le K} \left[d_{uu^{\App} _i}-d_{uu^{\App} _j}\right]^2
            \indicator{\pi^*(u^\App _i) \ge \pi^*(u^\App _i)} } 
\le  \sqrt{2 {K \choose 2}}\sqrt{ K  \sum_{i\in[K]} \left(d_{uu^{\App} _i} - d_{uu^*_i}\right)^2  }
\end{align}
\end{proof}

\begin{lemma}\label{lem:util-area-dp}
Define $\epsilon=\sqrt{\frac {2\log (2/\delta)}{|\Vcal|}}+\frac{7 \log (2/\delta)}{3(|\Vcal|-1)}$.
Then, with probability at least $1-2K\delta$, we have the following bounds:
\begin{align}
&(i)\ \App\text{ on }\tCN: \ 0\le\sum_{i\in[K]} \Area (u  , u^* _i)-\sum_{i\in[K]} \Area (u,u^{\overline{\tCN}} _i) \le (2K\epsilon+\utildpCN/|\Vcal|),\nn\\
&(ii)\ \App\text{ on }\tAA: \ 0\le\sum_{i\in[K]} \Area (u  , u^* _i)-\sum_{i\in[K]} \Area (u,u^{\overline{\tAA}} _i) \le   \log (|\Vcal|\Omega(r))(2K\epsilon+\utildpAA/|\Vcal|),\nn\\
&(iii)\ \App\text{ on }\tJC: \ 0\le\sum_{i\in[K]} \Area (u  , u^* _i)-\sum_{i\in[K]} \Area (u,u^{\overline{\tJC}} _i) \le  2 \Omega(r)(2K\epsilon+\utildpJC/|\Vcal|)\nn.
\end{align}
Recall that   $\Area(u,v)$ is the common volume
between $D$-dimensional hyperspheres centered around $u$ and $v$.
\end{lemma}
\begin{proof}
\noindent \boxed{\text{Proof of (i).}} We observe that:
\begin{align}
& \sum_{i\in[K]}\Area (u,u^*_i)-\sum_{i\in[K]}\Area (u,u^{\overline{CN}}_i) \geq  \frac{ \sum_{i\in[K]} s_{\tCN}(u, u^{CN} _i)
-\sum_{i\in[K]} s_{\tCN}(u, u^{\overline{CN}} _i)}{ |\Vcal| }+ 2 k\epsilon\nn\\[-0.05cm] 
& \overset{1}{\implies}  \sum_{i\in[K]} \frac{s_{\tCN}(u, u^{CN} _i)}{|\Vcal|} -\sum_{i\in[K]}\frac{s_{\tCN}(u, u^{*} _i)}{ |\Vcal| }
+  \sum_{i\in[K]} \frac{s_{\tCN}(u, u^{\overline{CN}} _i)}{ |\Vcal| } -\sum_{i\in[K]}\frac{s_{\tCN}(u, u^{\tCN} _i)}{|\Vcal|}
\end{align}
\begin{align}
&\quad\quad+\sum_{i\in[K]}\Area (u,u^*_i)-\sum_{i\in[K]}\Area (u,u^{\overline{CN}} _i)  \geq 2 k\epsilon\nn\\[-0.05cm]
&\overset{2}{\implies}  \bigvee_{i\in[K]}\left (  \frac{s_{\tCN}(u, u^{\overline{CN}} _i )}{|\Vcal|} - \Area (u, u^{\overline{CN}} _i)
\geq \epsilon \right )  \bigvee_{i\in[K]}\left (  -\frac{s_{\tCN}(u, u^{*} _i)}{|\Vcal|} + \Area (u,u^* _i) \geq \epsilon \right )\nn\\[-0.05cm]
&\overset{3}{\implies}  \PP\left(  \sum_{i\in[K]}\Area (u,u^*_i)-\sum_{i\in[K]}\Area (u, u^{\overline{CN}} _i) \geq  \frac{-\sum_{i\in[K]} s_{\tCN}(u,  u^{\overline{CN}} _i)
+\sum_{i\in[K]} s_{\tCN}(u, u^\tCN _i)}{ |\Vcal| }+ 2 k\epsilon  \right)\nn\\[-0.05cm]
& \le \sum_{i\in[K]} \PP  \left (  \frac{s_{\tCN}(u,  u^{\overline{CN}} _i)}{|\Vcal|}
- \Area (u, u^{\overline{CN}} _i) \geq \epsilon \right ) + \sum_{i\in[K]} \PP \left (  -\frac{s_{\tCN}(u, u^{*} _i)}{|\Vcal|} + \Area (u,u^* _i) \geq \epsilon \right ) \nn\\[-0.05cm]
&\overset{4}{\le} 2K\delta
\end{align}
The statement (1) is because $s_{\tCN} (u, u^{\tCN}  _i)>s_{\tCN} (u  ,u^{*}  _i)$. The statement (2) is because $X+Y>a\implies X>a \text{ or } Y>a$. The statement (3) is due to (1) and (2). Ineq. (4) is due to empirical Bernstein inequality~\cite{maurer2009empirical}.

\noindent \boxed{\text{Proof of (ii).}} Note that:
$\EE[s_\tAA (u,v)]=\frac{\Area(u,v)}{\log(|\Vcal| \Omega(r))}$. 
%
We observe that:
\begin{align}
& \sum_{i\in[K]}\Area(u,u^*_i)-\sum_{i\in[K]}\Area(u, u^{\overline{AA}}_i) \geq \log(|\Vcal| \Omega(r)) \left[\frac{ \sum_{i\in[K]} s_{\tAA}(u, u^{\tAA} _i)
-\sum_{i\in[K]} s_{\tAA}(u, u^{\overline{\tAA}} _i)}{ |\Vcal| } + 2 K\epsilon\right] \nn\\[-0.05cm]
& \overset{1}{\implies}  \sum_{i\in[K]} \frac{s_{\tAA}(u, u^{\tAA} _i)}{|\Vcal|} -\sum_{i\in[K]}\frac{s_{\tAA}(u, u^{*} _i)}{ |\Vcal| }
+  \sum_{i\in[K]} \frac{s_{\tAA}(u, u^{\overline{\tAA}} _i)}{ |\Vcal| } -\sum_{i\in[K]}\frac{s_{\tAA}(u, u^{\tAA} _i)}{|\Vcal|}\nn\\[-0.05cm]
&\quad\quad+\sum_{i\in[K]}\frac{\Area (u,u^*_i)}{ \log(|\Vcal| \Omega(r)) }-\sum_{i\in[K]}\frac{\Area (u,u^{\overline{\tAA}} _i)}{ \log(|\Vcal| \Omega(r)) }  \geq 2 k\epsilon\nn\\
&\overset{2}{\implies}  \bigvee_{i\in[K]}\left (  \frac{s_\tAA (u, u^{\overline{AA}} _i )}{|\Vcal|} -  \frac{\Area (u,u^{\overline{AA}} _i)}{\log (|\Vcal|\Omega(r))} \geq \epsilon \right )
\bigvee_{i\in[K]}\left (  -\frac{s_\tAA (u, u^{*}  _i)}{|\Vcal|} +  \frac{\Area(u,u^*_i)}{\log (|\Vcal|\Omega(r))} \geq \epsilon \right )\nn\\[-0.05cm]
&\overset{3}{\implies}  \PP\left(  \sum_{i\in[K]}\Area(u,u^*_i)-\sum_{i\in[K]}\Area(u,u^{\overline{\tAA}} _i) \right.\nn\\[-0.05cm]
&\quad\quad\quad\quad \quad\quad  \geq \left. \log(|\Vcal| \Omega(r)) \left[\frac{ \sum_{i\in[K]} s_{\tAA}(u, u^{\tAA} _i) 
-\sum_{i\in[K]} s_{\tAA}(u, u^{\overline{\tAA}} _i)}{ |\Vcal| } + 2 K\epsilon\right] \right)\nn\\
& \le \sum_{i\in[K]} \PP\left( \frac{s_{\tAA}(u,  u^{\overline{\tAA}} _i)}{|\Vcal|} -  \frac{\Area(u,u^{\overline{\tAA}} _i)}
{\log (|\Vcal|\Omega(r))} \geq \epsilon \right)+ \sum_{i\in[K]} \PP\left(  -\frac{s_\tAA (u, u^{*}  _i)}{|\Vcal|} + \frac{\Area(u,u^*_i)}{\log (|\Vcal|\Omega(r))} \geq \epsilon  \right)\nn\\[-0.05cm]
&\overset{4}{\le} 2k\delta
\end{align}
The statements (1)---(4) follow using the same argument as in the proof of~(i).

\noindent \boxed{\text{Proof of (iii).}} Note that:
$\EE[s_{\tJC} (u,v)]=\frac{A(u,v)}{2\Omega(r)-A(u,v)} $.
Suppose we have:
\begin{align}
& \sum_{i\in[K]}\Area(u,u^*_i)-\sum_{i\in[K]}\Area(u,u^{\overline{\tJC}} _i) \geq  2\Omega(r) \left[\frac{ \sum_{i\in[K]} s_{\tJC}(u, u^{\tJC} _i) 
-\sum_{i\in[K]} s_{\tJC}(u, u^{\overline{\tJC}} _i)}{ |\Vcal| } + 2 K\epsilon\right]\nn
\end{align}
which implies that
\begin{align}
&\sum_{i\in[K]}\Area(u,u^*_i)-\sum_{i\in[K]}\Area(u,u^{\overline{\tJC}} _i)\nn\\
& \  \geq \left[\frac{ \sum_{i\in[K]} s_{\tJC}(u, u^{\tJC} _i) 
-\sum_{i\in[K]} s_{\tJC}(u, u^{\overline{\tJC}} _i)}{ |\Vcal| } + 2 K\epsilon\right] \left[\frac{(2\Omega(r)-\Area(u,u^* _i))(2\Omega(r)-\Area(u, u^{\overline{\tJC}} _i))} {2\Omega(r)}\right]\nn\\[-0.05cm]
& \overset{1}{\implies}  \sum_{i\in[K]} \frac{s_{\tJC}(u, u^{\tJC} _i)}{|\Vcal|} -\sum_{i\in[K]}\frac{s_{\tJC}(u, u^{*} _i)}{ |\Vcal| }
+  \sum_{i\in[K]} \frac{s_{\tJC}(u, u^{\overline{\tJC}} _i)}{ |\Vcal| } -\sum_{i\in[K]}\frac{s_{\tJC}(u, u^{\tJC} _i)}{|\Vcal|}\nn
\end{align}
\begin{align}
&\quad\quad+\sum_{i\in[K]}\frac{\Area (u,u^*_i)}{ 2\Omega(r)-\Area(u,u^*_i) }-\sum_{i\in[K]}\frac{\Area (u,u^{\overline{\tJC}} _i)}{ 2\Omega(r)-\Area(u, u^{\overline{\tJC}} _i) }  \geq 2 k\epsilon\nn\\[-0.05cm]
&\overset{2}{\implies}  \bigvee_{i\in[K]}\left (  \frac{s_\tJC (u, u^{\overline{\tJC}} _i )}{|\Vcal|} -  \frac{\Area (u,u^{\overline{\tJC}} _i)}{2\Omega(r)-\Area(u,u^{\overline{\tJC}} _i)} \geq \epsilon \right )
\bigvee_{i\in[K]}\left (  -\frac{s_\tJC (u, u^{*}  _i)}{|\Vcal|} +  \frac{\Area(u,u^*_i)}{2\Omega(r)-\Area(u,u^*_i)} \geq \epsilon \right )\nn\\[-0.05cm]
&\overset{3}{\implies}  \PP\left(  \sum_{i\in[K]}\Area(u,u^*_i)-\sum_{i\in[K]}\Area(u,u^{\overline{\tJC}} _i) \right. \nn\\[-0.05cm]
&\quad\quad\quad\quad \quad\quad  \geq \left. 2\Omega(r)  \left[\frac{ \sum_{i\in[K]} s_{\tJC}(u, u^{\tJC} _i) 
-\sum_{i\in[K]} s_{\tJC}(u, u^{\overline{\tJC}} _i)}{ |\Vcal| } + 2 K\epsilon\right] \right)\nn\\ 
& \le \sum_{i\in[K]} \PP\left( \frac{s_{\tJC}(u,  u^{\overline{\tJC}} _i)}{|\Vcal|} -  \frac{\Area(u,u^{\overline{\tJC}} _i)}{2\Omega(r)-\Area(u,u^{\overline{\tJC}} _i )} 
\geq \epsilon \right)+ \sum_{i\in[K]} \PP\left(  -\frac{s_\tJC (u, u^{*}  _i)}{|\Vcal|} + \frac{\Area(u,u^*_i)}{2\Omega(r)-\Area(u,v)} \geq \epsilon  \right)\nn\\[-0.05cm]
&\overset{4}{\le} 2k\delta
\end{align}
Statements (1)---(4) follow with the same argument as in the proof of~(i).
\end{proof}
\begin{lemma}\label{lem:d-univ}
For any ranking based LP algorithm $\Acal$, we have
\begin{align}
  \sum_{i=1} ^K (d_{u\uaa_i}-d_{u u^*_i})\le  2rK -\Area ^{-1}\left(\sum_{i=1} ^K \Area(u,u^*_i)-\sum_{i=1} ^K \Area(u,\uaa_i)\right),
 \end{align}
 where $\Area^{-1}(x)$ equals to the distance between nodes $u,v$ in $D$ dimensional hypersphere,  for which $\Area(u,v)=x$.
\end{lemma}
\begin{proof}
We have the following series of inequalities, each due to Proposition~\ref{aux-prop-1} (ii).
\begin{align}
&\Area ^{-1}\left(\Area (u, \uaa_1)\right)+ \Area ^{-1}\left(\Area (u, \uaa_2)\right) \nn\\
&\qquad \qquad \qquad \qquad{\le} \Area ^{-1}(0)+\Area ^{-1} \left(\Area (u, \uaa_1)+\Area (u, \uaa_2)\right) \\[0.3cm]
&\Area ^{-1}\left(\Area (u, \uaa_1)+\Area (u, \uaa_2)\right)+ \Area ^{-1}\left(\Area (u, \uaa_3)\right)\nn \\
& \qquad \qquad \qquad \qquad {\le} \Area ^{-1}(0)+\Area ^{-1} \left(\Area (u, \uaa_1)+\Area (u, \uaa_2)+\Area (u, \uaa_3)\right) \\
& \hspace{6cm}\vdots  \nn \\
&\Area ^{-1}\left(\sum_{i=1}^{K-1} \Area (u, \uaa_1)\right)+ 
\Area ^{-1}\left(\Area (u, \uaa_K)\right)\nn\\
& \qquad \qquad \qquad \qquad  {\le} \Area ^{-1}(0)+\Area ^{-1} \left(\sum_{i=1}^K \Area (u, \uaa_K)\right) \\[-0.05cm]
&\Area ^{-1}\bigg(\sum_{i=1}^K \Area(u,\uaa_i)\bigg) +\Area ^{-1} \bigg( \sum_{i=1}^K \Area (u,u^* _i)- \sum_{i=1}^K \Area (u, \uaa_i) \bigg)\nn \\
& \qquad \qquad \qquad \qquad  {\le} \Area ^{-1}(0)+\Area ^{-1} \bigg(\sum_{i=1} ^K \Area (u,u^* _i)\bigg)  
\end{align}
Taking the telescoping sum, we have:
\begin{align}
&\sum_{i=1}^K \Area ^{-1} \left(\Area (u, \uaa_i)\right)+\Area ^{-1} \left( \sum_{i=1}^K \Area (u,u^* _i)- \sum_{i=1}^K \Area (u, \uaa_i) \right) \nn\\
&\le K\Area^{-1}(0)+\Area ^{-1}\left(\sum_{i=1}^K \Area (u,u^* _i)\right)\\
& \le 2Kr+\Area ^{-1}\left(\sum_{i=1}^K \Area (u,u^* _i)\right)\qquad (\text{Due to }\Area^{-1}(0) =  2r) 
\end{align}
\begin{align}
&{\le} 2Kr+\Area ^{-1}\left(\Area (u,u^* _1)\right) \qquad \left(\Area ^{-1} (.) \text{is decreasing (Proposition~\ref{aux-prop-1} (i))} \right) \\
&{\le} 2Kr+\sum_{i=1}^K \Area ^{-1}\left( \Area (u,u^* _i)\right) \qquad (\text{adding additional positive terms})
\end{align}
Hence, 
$ \sum_{i=1} ^K (d_{u\uaa_i}-d_{u u^*_i})\le 2rK -\Area ^{-1}\bigg(\sum_{i=1} ^K \Area(u,u^*_i)-\sum_{i=1} ^K \Area(u,\uaa_i)\bigg)$
\end{proof}

Lemma~\ref{lem:d-univ} can be used to prove the corresponding bounds for different LR heuristics.
\begin{lemma}\label{lem:lastX}
 If we define  $\epsilon=\sqrt{\frac {2\log (2/\delta)}{|\Vcal|}}+\frac{7 \log (2/\delta)}{3(|\Vcal|-1)}$, then, with probability $1-2K\delta$
 \begin{align}
 & (i) \ \tCN: \ \ \sum_{i=1} ^K (d_{u u^{\dCN}_i}-d_{u u^*_i})\le 2Kr\left(\frac{2K {\epsilon}+\utildpCN/|\Vcal|}{\Omega(r)}\right)^{1/KD}\\
 & (ii)\ \tAA: \  \ \sum_{i=1} ^K (d_{uu^{\dAA}_i}-d_{u u^*_i})\le  2Kr\left(\frac{\log (|\Vcal| \Omega(r))( 2K {\epsilon} +\utildpAA/|\Vcal| )}{\Omega(r)}\right)^{1/KD}\\
 & (iii)\ \tJC: \ \ \sum_{i=1} ^K (d_{uu^{\dJC}_i}-d_{u u^*_i})\le 2Kr\left({4K\epsilon+2\utildpJC/|\Vcal|}\right)^{1/KD}
  \end{align}
\end{lemma}

\begin{proof}
From Lemma~\ref{lem:d-univ}, for $\Acal\in\set{\tAA,\tCN,\tJC}$, we have:
 \begin{align}
 \sum_{i=1} ^K (d_{u u^{\Acal} _i}-d_{u u^*_i})\le 2rK -\Area ^{-1}\bigg(\sum_{i=1} ^K \Area(u,u^*_i)-\sum_{i=1} ^K \Area(u,u^{\Acal} _i)\bigg)\label{eq:inm-1}
\end{align}
We define $\overline{\epsilon}= \dfrac{\sum_{i=1} ^K \Area(u,u^*_i)-\sum_{i=1} ^K \Area(u,u^{\Acal} _i)}{K}$.  
Then, Eq.~\eqref{eq:inm-1} is less than $ 2rK-\Area ^{-1} (2K\overline{\epsilon})$ with probability $1-2K\delta$ from Lemma~\ref{lem:d-univ}.
Now, 
\begin{align}
& \Area(u,v) \ge \Omega(r) \bigg(1-\frac{d_{uv}}{2r}\bigg)^{D}\ \  \forall d_{uv}>0\\
&\implies \Area ^{-1} (2K\overline{\epsilon}) \ge 2r\left(1-(2K\overline{\epsilon}/{\Omega(r)})^{1/D}\right)
\end{align}
We aim to approximate the above quantity by 
\begin{align}
 2r(1-(2K\overline{\epsilon}/{\Omega(r)})^{1/D}) \ge 2rK(1-(2K\overline{\epsilon}/{\Omega(r)})^{1/A})
\end{align}
for some suitable dimension $A$, which says that
\begin{align}
A \ge \frac{\log (2K\overline{\epsilon}/{\Omega(r)})} {\log\left(1-\frac{1}{K}+\frac{1}{K} \left(\frac{2K\overline{\epsilon}}{\Omega(r)}\right)^{1/D}\right) } 
\ge \frac{\log (2K\overline{\epsilon}/{\Omega(r)})} {\log\left(1-\frac{1}{K}+\frac{1}{K} \left(\frac{2K\overline{\epsilon}}{\Omega(r)}\right)^{1/D}\right) }
 {\ge} \frac{\log (2K\overline{\epsilon}/{\Omega(r)})} {1-\frac{1}{K}+\frac{1}{KD} \log \left(\frac{2K\overline{\epsilon}}{\Omega(r)}\right) }\nn
\end{align}
The last inequality is due to concavity of logarithmic function. Now, 
the quantity $\frac{\log (2K\overline{\epsilon}/{\Omega(r)})} {1-\frac{1}{K}+\frac{1}{KD} \log \left(\frac{2K\overline{\epsilon}}{\Omega(r)}\right) }$ achieves
maximum at $\overline{\epsilon}\to 0$ when $A\ge KD$. Then, $\sum_{i=1} ^K (d_{u\uaa_i}-d_{u u^*_i}) \le 2Kr(2K\overline{\epsilon}/{\Omega(r)})^{1/KD}$.

Lemma~\ref{lem:util-area-dp} provides the following:
  \vspace{0.2cm}
\begin{itemize}
    \item By putting $\overline{\epsilon}=\epsilon+\utildpCN/(2K|\Vcal|)$, we obtain the bound for $\Acal=\tCN$.
    \vspace{0.2cm}
    \item By putting $\overline{\epsilon}= \log (|\Vcal| \Omega(r))( 2K {\epsilon} +\utildpAA/|\Vcal| )$, we obtain the bound for $\Acal=\tAA$.
    \vspace{0.2cm}
    \item By putting $\overline{\epsilon}=\Omega(r)\cdot\left(4K\epsilon+2\utildpJC/|\Vcal|\right)$, we obtain the bound for $\Acal=\tJC$.
\end{itemize}

\end{proof}

\begin{proposition}\label{aux-prop-1}  
(i)~$\Area ^{-1} (y)$ is decreasing and convex.
(ii)~$\Area^{-1}  (x)+ \Area^{-1}  (a-x)< \Area^{-1}  (0)+ \Area^{-1}  (a)$. 
\end{proposition}
\begin{proof}
(i)~We have that
 \begin{align}
  \frac{d \Area ^{-1}(y)}{dy}=-\left[C \left(1-\frac{( \Area ^{-1} (y))^2}{4r^2}\right)^{\frac{D-1}{2}}\right]^{-1},
 \end{align}
 where $C>0$ depends on $D$ and $r$~\cite{SarkarCM2011LPembed}.
We differentiate again and have
 \begin{align}
  \frac{d^2 \Area ^{-1}(y)}{dy^2}=-\left[C \left(1-\frac{( \Area ^{-1} (y))^2}{4r^2}\right)^{\frac{D+1}{2}}\right]^{-1} \Area^{-1}  (y) \frac{d \Area ^{-1}(y)}{dy}>0
 \end{align}
 (ii) Assume $f(x)= \Area^{-1}  (x)+ \Area^{-1}  (a-x)$. $d^2f(x)/dx^2= \frac{d^2 \Area^{-1}  (x)}{dx^2}+ \frac{d^2 \Area^{-1}  (a-x)}{dx^2}>0$.
 Moreover, $f'(x)=0$ at $x=a/2$. Hence $f(x)$ is U-shaped convex function. So, $f(x)\le f(0)=f(a)$.
\end{proof}

\section{Additional details about experiments}
\label{sec:app:expt-details}

\subsection{Dataset details}\label{sec:app:dataset}
We use five diverse datasets for our experiments. 
\begin{itemize}
\item \textbf{Facebook}~\cite{leskovec2012learning} is a snapshot of a part of Facebook's social network.
\item \textbf{USAir}~\cite{usair} is a network of US Air lines.
\item \textbf{Twitter}~\cite{leskovec2012learning} is a snapshot of a part of Twitter's social network.
%
\item \textbf{Yeast}~\cite{von2002comparative} is a protein-protein interaction
network in yeast. 
\item \textbf{PB}~\cite{ackland2005mapping} is a network of US political blogs.
\end{itemize}
\begin{table*}[!ht]
\centering
\resizebox{0.6\textwidth}{!}{
\begin{tabular}{|l||c|c|c|c|c|}
\hline

Dataset &$|\Vcal|$&$|\Ecal|$& $d_{avg}$ & Clust. Coeff. & Diameter \\ \hline \hline
Facebook	 & 4039 	 & 	 88234 	 & 	 43.69& 	 0.519& 	 8 \\ \hline 
USAir	 & 332 	 & 	 2126 	 & 	 12.81& 	 0.396 	& 60\\ \hline 
Twitter&	235	 & 10862	&	 92.44 &	 0.652 &  3\\ \hline 
Yeast	 & 2375 	 & 	 11693 	 & 	 9.85& 	 0.469 	 &15\\ \hline
PB	 & 1222 	 & 	 16714 	 & 	 27.36& 	 0.226& 	 8\\ \hline 
\end{tabular} }
\caption{Dataset statistics.}
\label{tab:stat-dataset}
\end{table*}%

\subsection{Implementation details about \our\ and its variants}
We design $\tf_{\thetab}$ (Eq.~\eqref{eq:dplp-linear} and~Eq.~\eqref{eq:umnn}) with deep neural networks.  
In practice, we set $a_j = \frac{1}{2}+\frac{j-1}{100}$ and $n_a=170$ 
to compute the intermediate quantity $\nu_{\betab}(s_{\Acal}(\cdot,\cdot))$ in~Eq.~\eqref{eq:dplp-linear}.
Our integrand function $g_{\bm{\phi}}(.)$ in Eq.~\eqref{eq:umnn} consists of one input layer, 20 hidden layers and
one output layer, in which each of the input  and 
hidden layers is a cascaded network of one linear and ReLU units and the output layer
is an ELU activation unit. We set the margin $\varrho=0.1$ in the pairwise loss in Eq.~\eqref{eq:opt-problem}.

We uploaded the code in \url{https://rebrand.ly/dplp} and also provided it in the supplementary material.
We implemented our methods--- \our\ and its two variants \ie, \our-Lin and \our-UMNN--- using Adam optimizer in pytorch.
In all cases, we set learning rate to be $0.1$ and weight decay to be $10^{-5}$. We used the neighbors and non-neighbors
for each query as an individual batch.

\subsection{Base LP protocols and their implementation details}\label{sec:app:impl}

\xhdr{Base LP protocols}
We  provided results for AA, CN, GCN and Node2Vec as four candidates of $\Acal$ in the main part of the paper.
Here, we also consider 
(i)~ two additional algorithms based on the triad-completion principle~\cite{LibenNowellK2007LinkPred} viz., 
Preferential Attachment (PA) and Jaccard coefficient (JC); and,  
(ii)~ four additional algorithms based on fitting node embeddings, viz., 
Struct2Vec~\cite{ribeiro2017struc2vec}, DeepWalk~\cite{perozzi2014deepwalk}, LINE~\cite{tang2015line} and PRUNE~\cite{lai2017prune}. 

\xhdr{Implementation details}
For triad-based LP heuristics, the implementations are trivial, and we need no hyper parameter tuning.
For embedding based methods, the node representations are generated using the training graph.
Then, following earlier work \cite{ZhangC2018LinkPredGNN,grover2016node2vec}, we use the $L_2$  features and use them to train a logistic classifier (in Liblinear~\cite{fan2008liblinear}) and then use the trained model to predict links.  In each of these method, we set the dimension of embedding to be~$z=80$ cross-validation.
The remaining hyper-parameters are as follows, which is also set using cross-validation.

 \begin{itemize}
    \item \textbf{GCN}: Number of epochs is $80$, learning rate is $0.01$, weight decay  is $5\times 10^{-4}$ and the number of hidden layers is $20$.
    \item \textbf{Node2Vec}: Number of walks is $10$ and the walk length is $80$.
    \item \textbf{PRUNE}: Learning rate is $10^{-4}$, the Number of epochs is $50$ and batch size is $300$.
    \item \textbf{DeepWalk}: Learning rate is $10^{-4}$, the Number of epochs is $50$ and batch size is $300$.
    \item \textbf{LINE}: Learning rate is $10^{-4}$, and  the number of negative samples is $5$.
    \item \textbf{Struct2Vec}: Number of walks is $10$, the walk length is $10$ and the number of layers is $6$.
\end{itemize}
 All the other hyperparameters are set as default in the corresponding software.

\subsection{Implementation details about baseline DP algorithms}
\xhdr{Staircase} Given the privacy leakage $\epsilon_p$, 
we produced noise from Staircase distribution~\cite[Algorithm 1]{geng2015optimal} \ie, \textsc{Staircase}$(\Delta_{\Acal},\epsilon_p)$, then added them to the scores and sorted the noised scores to select top-$K$ nodes.

\xhdr{Laplace} 
Given the privacy leakage $\epsilon_p$, 
we produced noise from Laplace distribution~\citep[Defs.\ 5, 6]{MachKDS2011AccuratePrivate} \ie, \textsc{Laplace}$(\Delta_{\Acal}/\epsilon_p)$, then added them to the scores and sorted the noised scores in decreasing order to select top-$K$ nodes.
 
\xhdr{Exponential}
Given the privacy leakage $\epsilon_p$,  we picked up the node  $v$  w.p. proportional to $\exp( \epsilon_p s_{\mathcal{A}} (u,v)/(2\Delta_{\mathcal{A}}) )$. 

\subsection{Computing infrastructure} 
The experiments were carried out on a 32-core 3.30GHz Intel server with 256GB RAM and an RTX2080 and a TitanXP GPU with 12GB GPU RAM each.

\section{Additional results}\label{sec:app:add}
\begin{figure*}
\centering
{\includegraphics[width=0.75\textwidth]{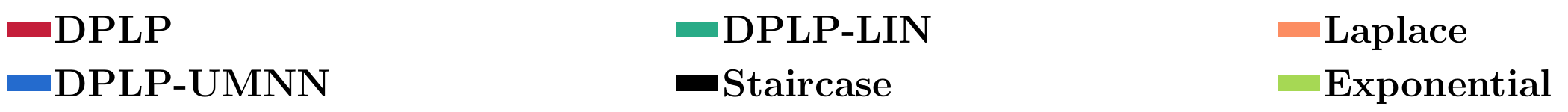}}
\\[-0.3cm]
\subfloat{\includegraphics[height=0.13\textwidth]{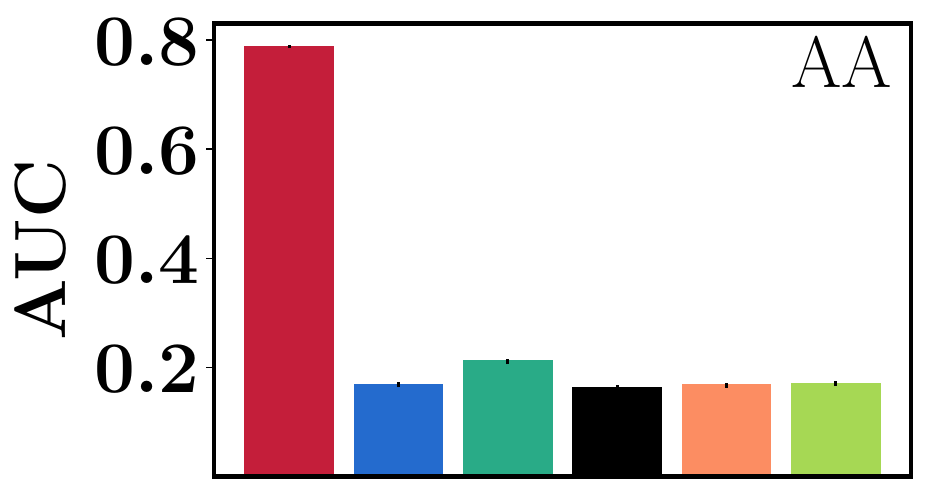}} 
\subfloat{\includegraphics[width=0.19\textwidth]{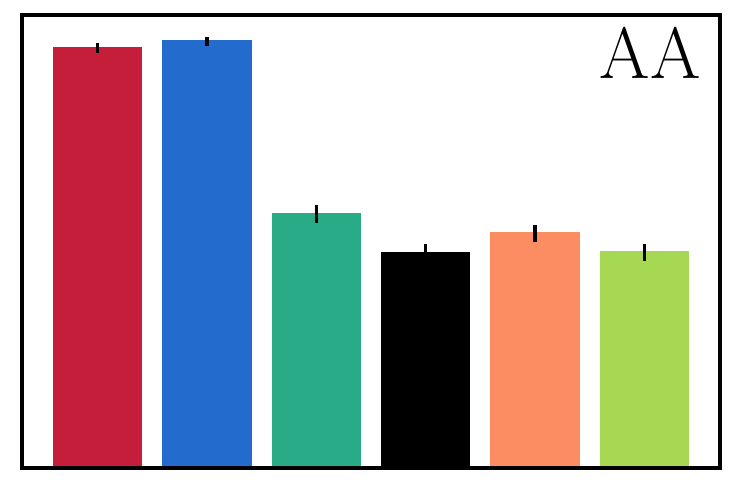}} 
\subfloat{\includegraphics[width=0.19\textwidth]{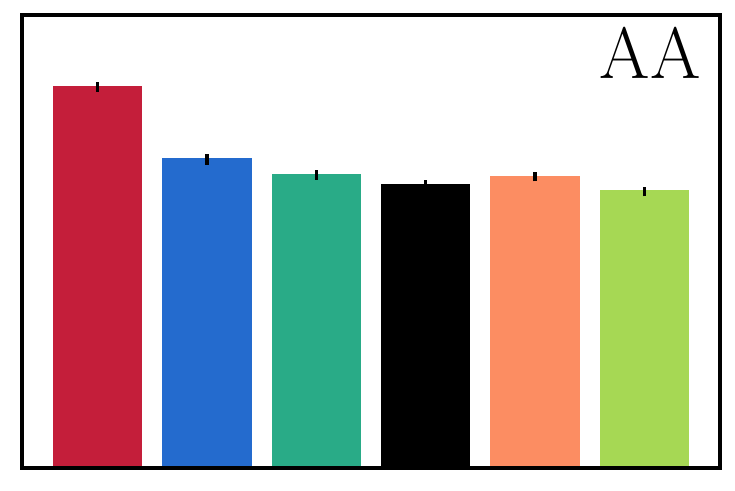}} 
\subfloat{\includegraphics[width=0.19\textwidth]{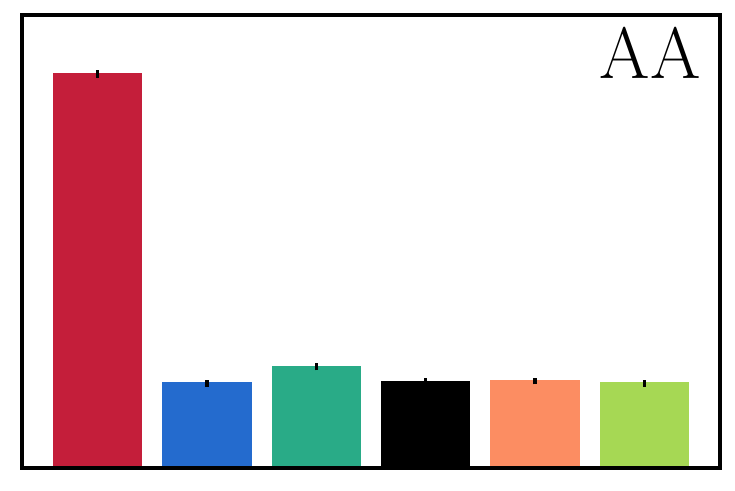}} 
\subfloat{\includegraphics[width=0.19\textwidth]{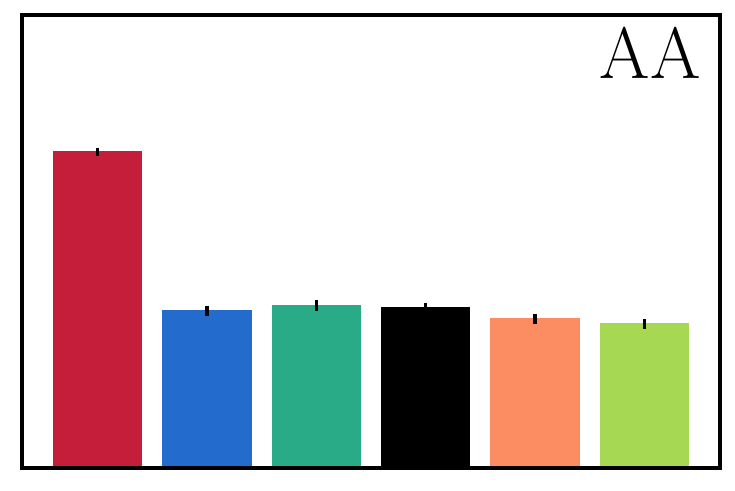}} \\
\subfloat{\includegraphics[height=0.13\textwidth]{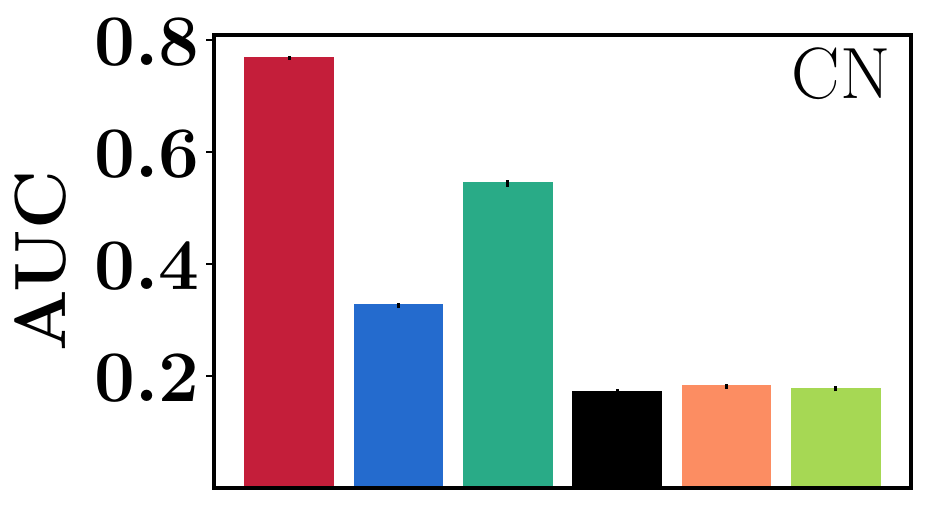}} 
\subfloat{\includegraphics[width=0.19\textwidth]{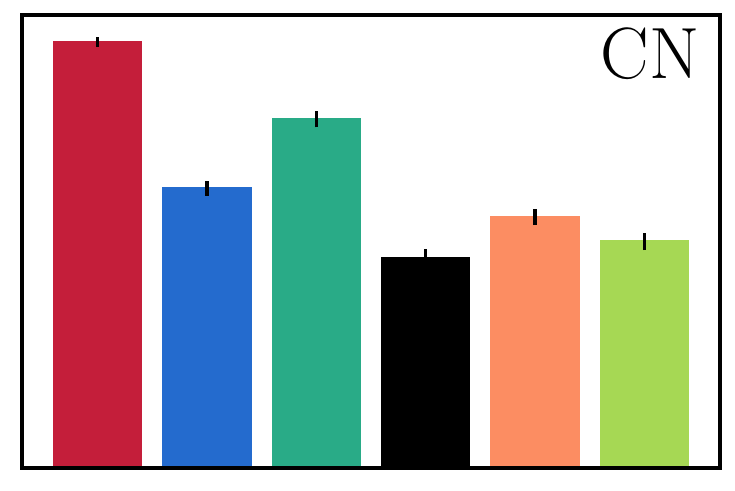}} 
\subfloat{\includegraphics[width=0.19\textwidth]{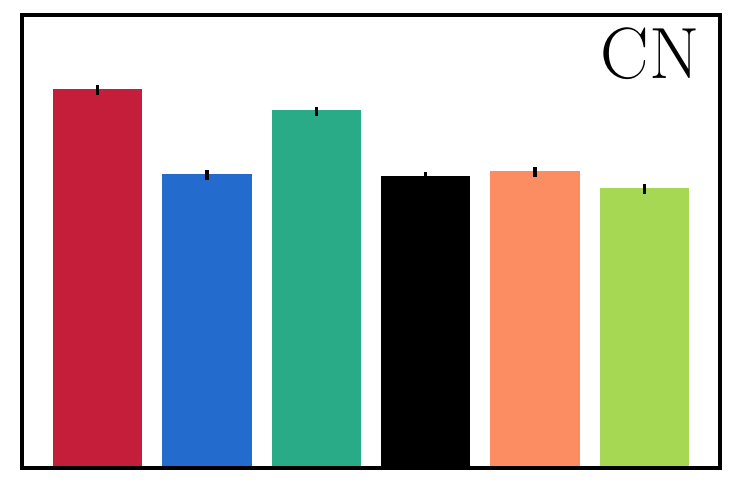}} 
\subfloat{\includegraphics[width=0.19\textwidth]{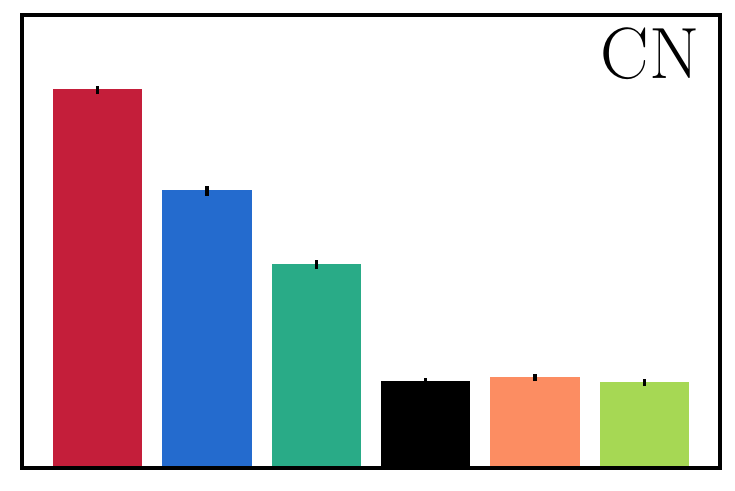}} 
\subfloat{\includegraphics[width=0.19\textwidth]{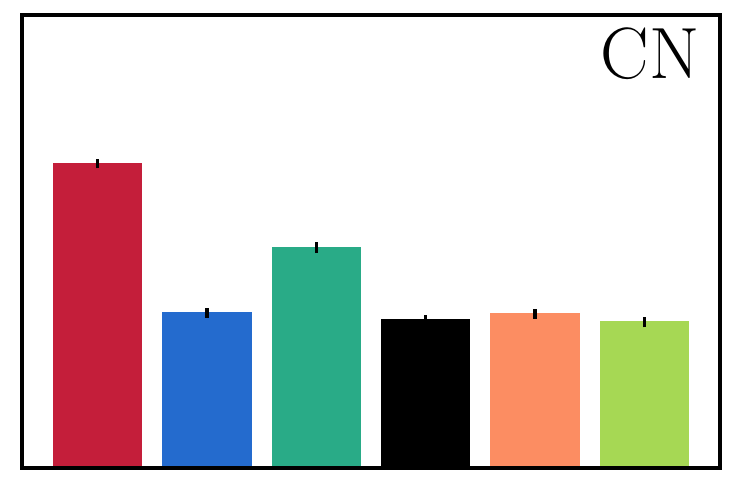}} \\
\subfloat{\includegraphics[height=0.13\textwidth]{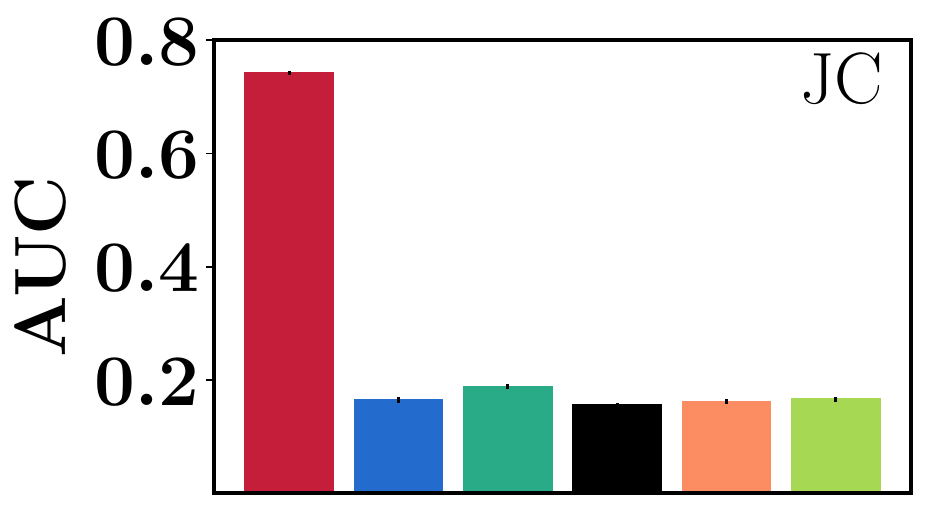}} 
\subfloat{\includegraphics[width=0.19\textwidth]{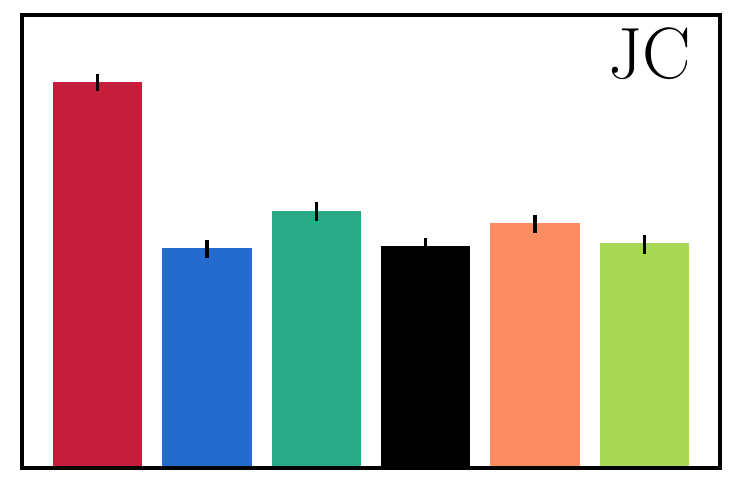}} 
\subfloat{\includegraphics[width=0.19\textwidth]{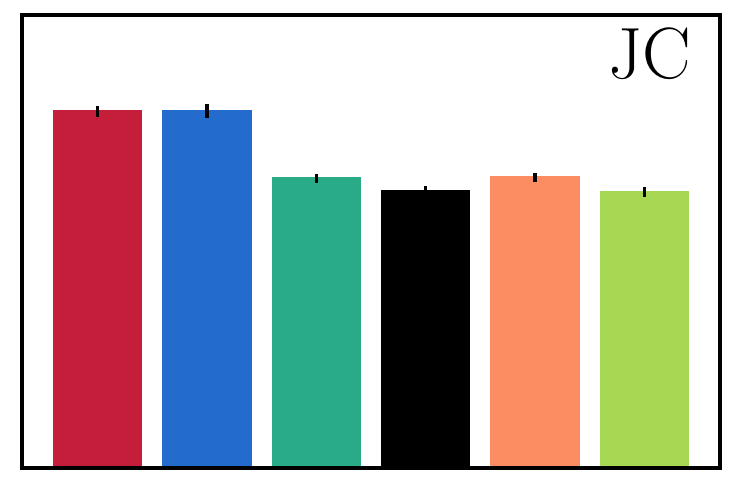}} 
\subfloat{\includegraphics[width=0.19\textwidth]{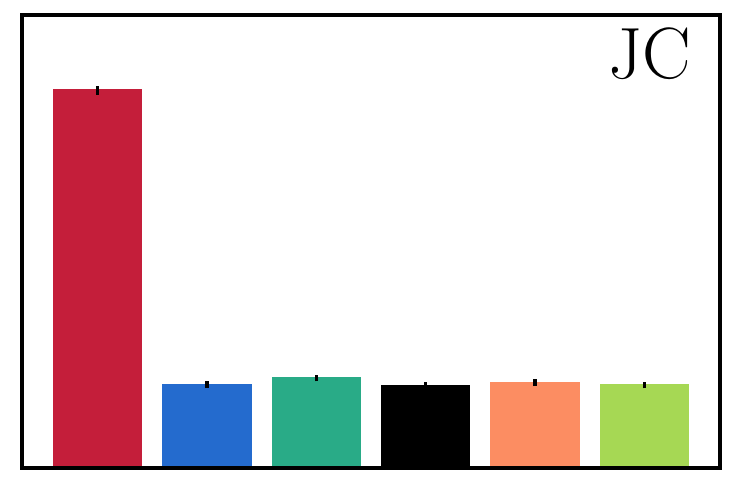}} 
\subfloat{\includegraphics[width=0.19\textwidth]{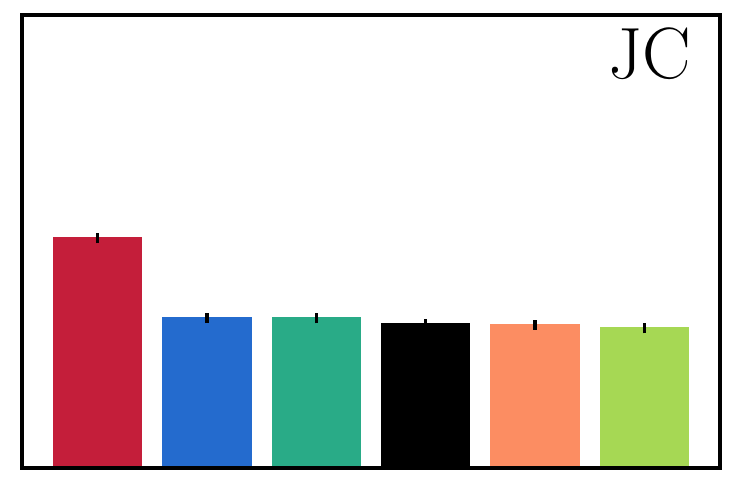}}\\
\subfloat[Facebook]{\setcounter{subfigure}{1}\includegraphics[height=0.13\textwidth]{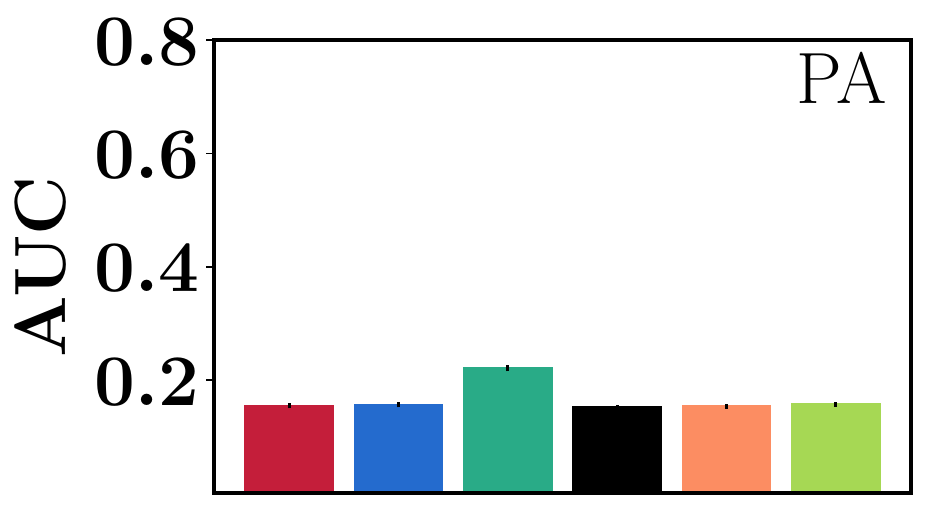}} 
\subfloat[USAir]{\includegraphics[width=0.19\textwidth]{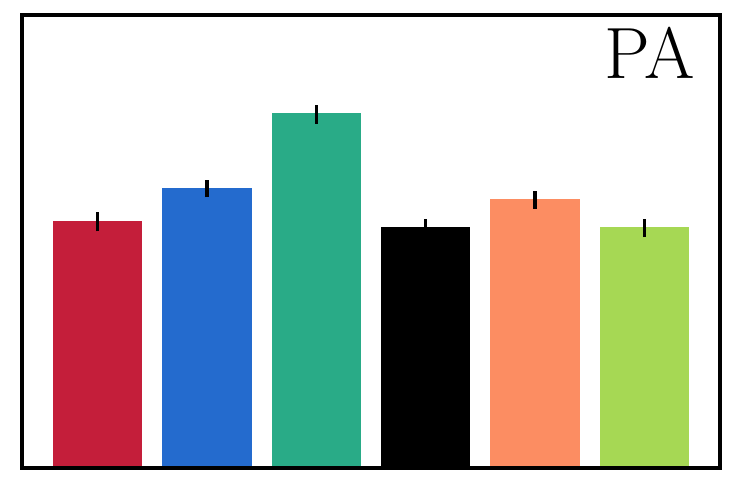}} 
\subfloat[Twitter]{\includegraphics[width=0.19\textwidth]{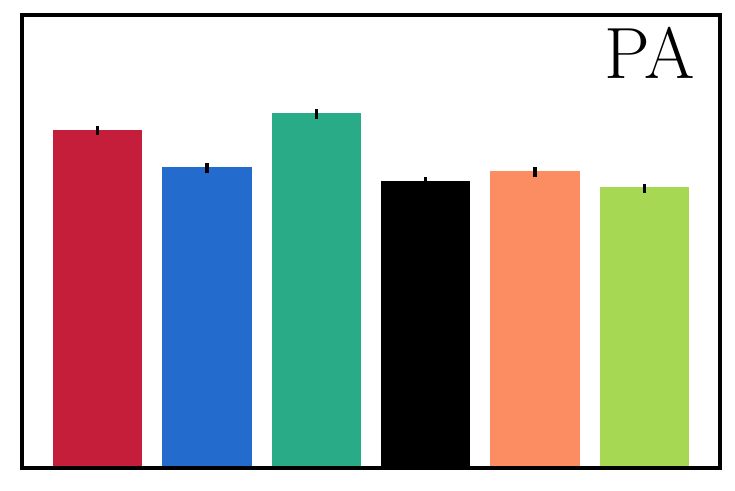}} 
\subfloat[Yeast]{\includegraphics[width=0.19\textwidth]{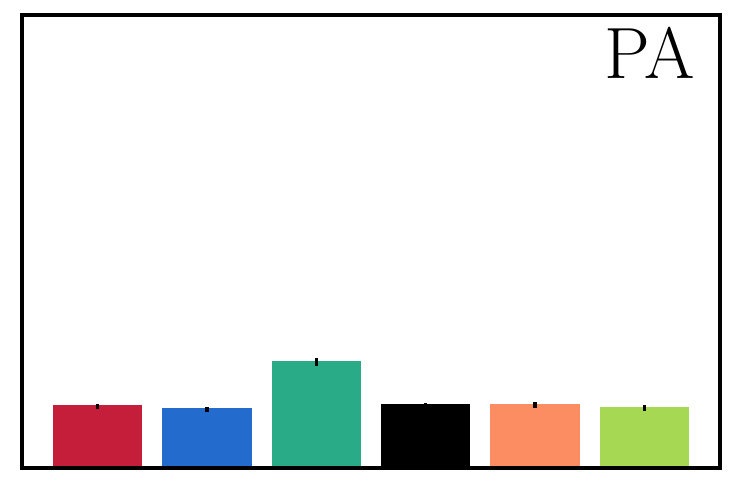}} 
\subfloat[PB]{\includegraphics[width=0.19\textwidth]{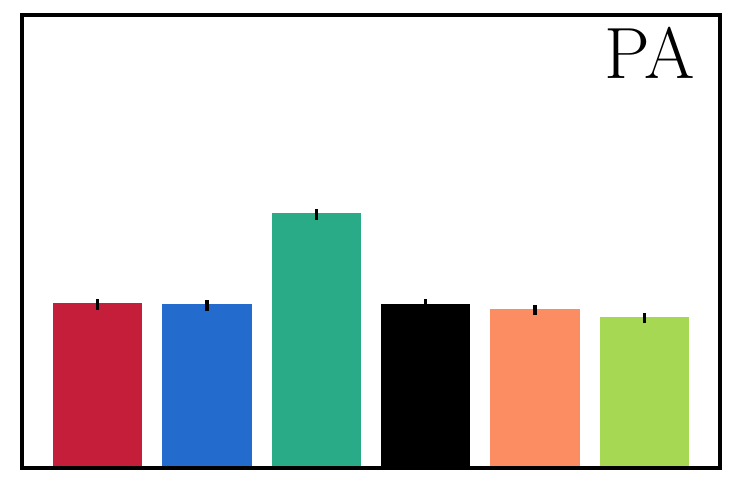}} 
\caption{Performance (AUC) on {Facebook, USAir, Twitter, Yeast, PB datasets} using various DP protocols: \our, \our-UMNN, \our-Lin, Staircase, Laplace and Exponential for 20\% held-out set with $\epsilon_p=0.1$ and $K=30$
for triad based LP protocols (AA, CN, JC, PA).
\our\ outperforms all the baselines in all cases, except preferential attachment, where \our-Lin turns out be the consistent winner.}
\label{fig:Add1}
\end{figure*}

\subsection{Results on triad based LP protocols}

Here, we first compare the predictive performance of \our{} against its two variants (\our-UMNN and \our-Lin) and three state-of-the-art baselines (Staircase, Laplace and Exponential) for 20\% held-out set with $\epsilon_p=0.1$ and $K=30$
for triad based LP protocols across all datasets.
Figure~\ref{fig:Add1} summarizes the results, which shows that \our\ outperforms all the baselines in all cases, except preferential attachment, where \our-Lin turns out be the consistent winner.

\subsection{Results on deep embedding based LP protocols}

Next, we compare the predictive performance of \our{} against its two variants (\our-UMNN and \our-Lin) and three state-of-the-art baselines (Staircase, Laplace and Exponential) for 20\% held-out set with $\epsilon_p=0.1$ and $K=30$,
for deep embedding based LP protocols across all datasets.
Figure~\ref{fig:Add2} summarizes the results, which shows that \our\ outperforms the baselines in majority of the embedding methods.
For our data sets, we observed that deep node embedding approaches generally fell short of traditional triad-based heuristics.
We believe iterated neighborhood pooling already introduces sufficient noise that relaxing $\epsilon_p$ does not help these node embedding methods.

\begin{figure*}[t!]
{\hspace{3cm}{\includegraphics[width=0.75\textwidth]{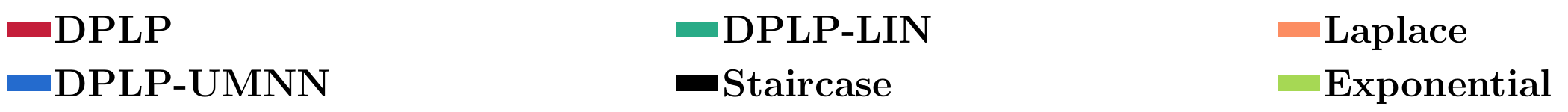}}
\\
\subfloat{\includegraphics[height=0.13\textwidth]{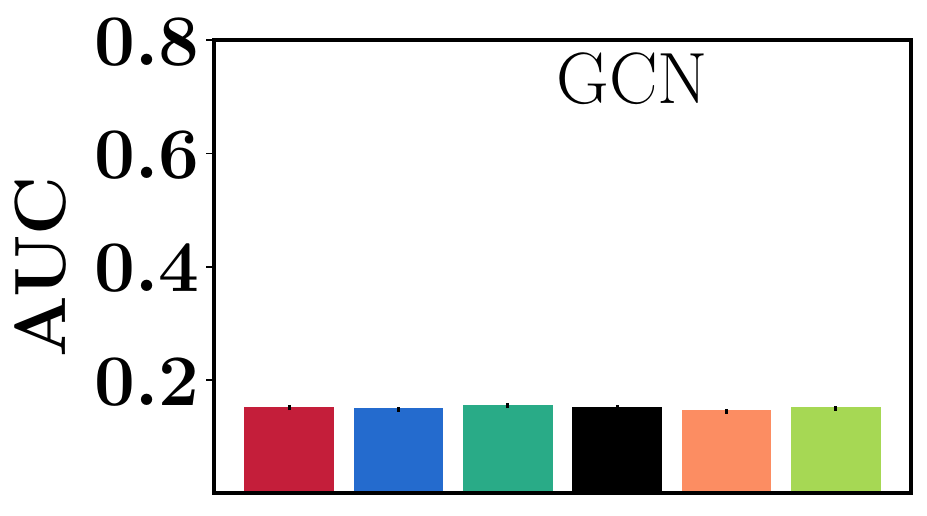}} 
\subfloat{\includegraphics[width=0.19\textwidth]{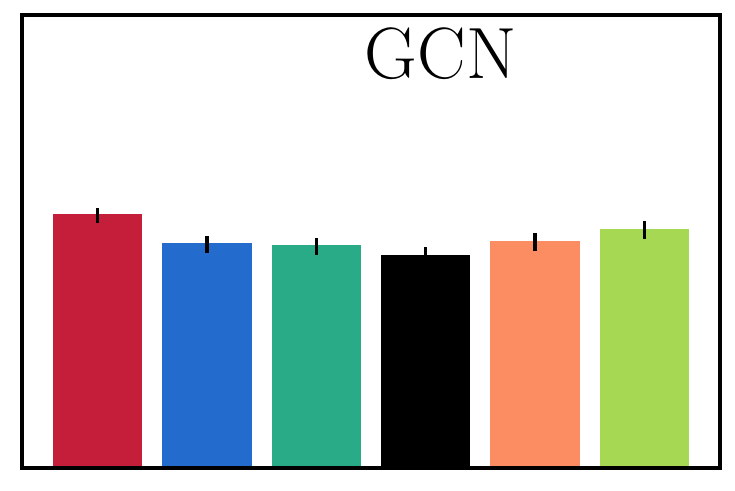}} 
\subfloat{\includegraphics[width=0.19\textwidth]{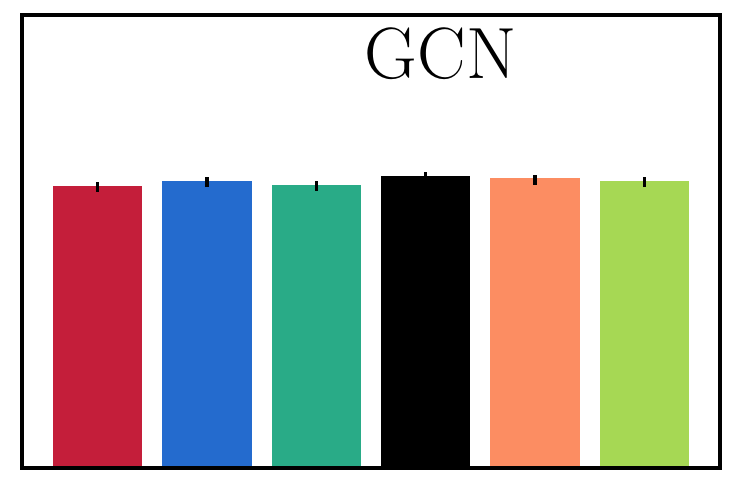}} 
\subfloat{\includegraphics[width=0.19\textwidth]{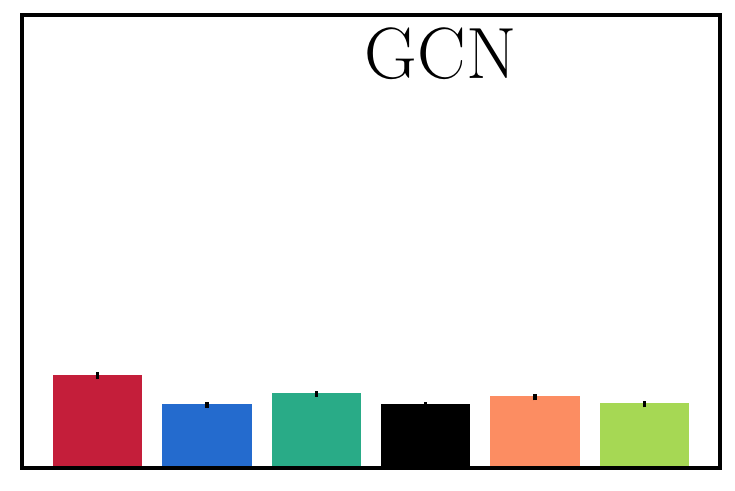}} 
\subfloat{\includegraphics[width=0.19\textwidth]{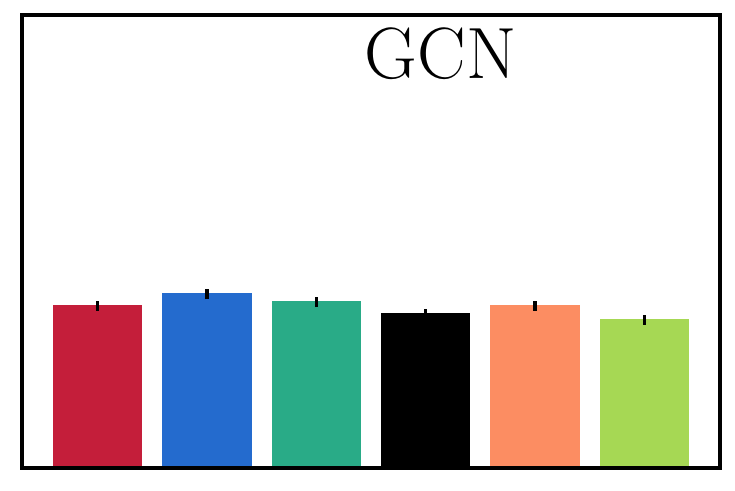}} \\
\subfloat{\includegraphics[height=0.13\textwidth]{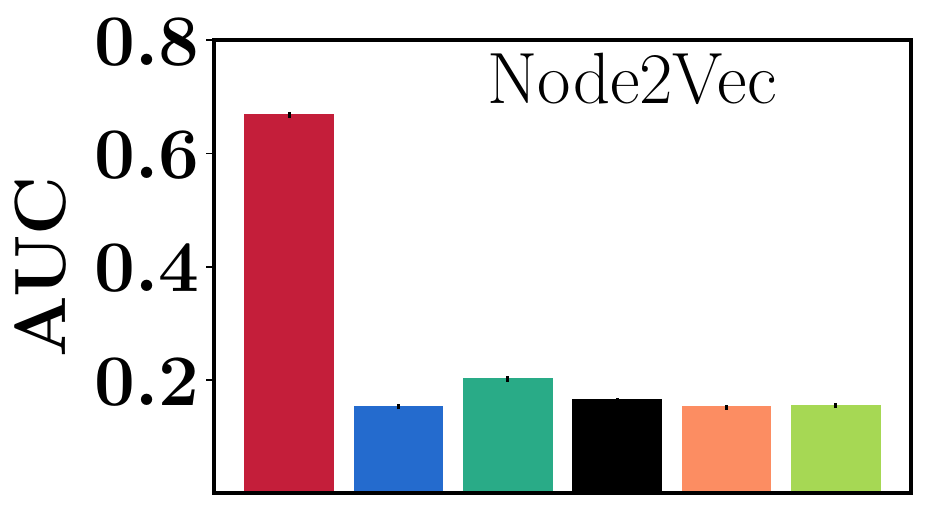}} 
\subfloat{\includegraphics[width=0.19\textwidth]{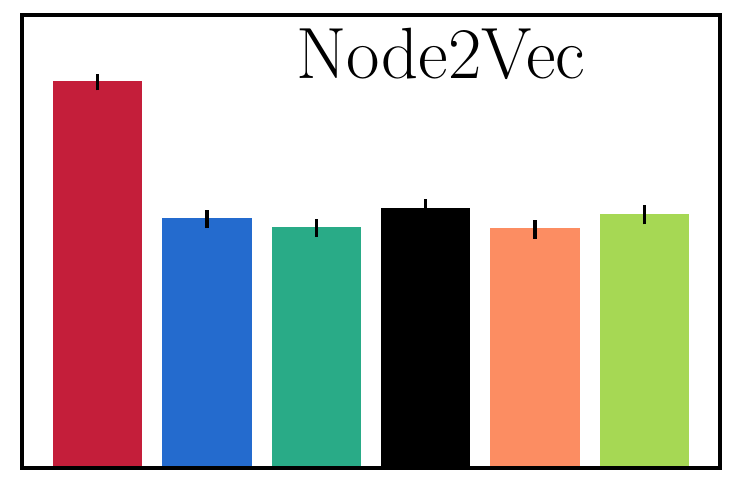}} 
\subfloat{\includegraphics[width=0.19\textwidth]{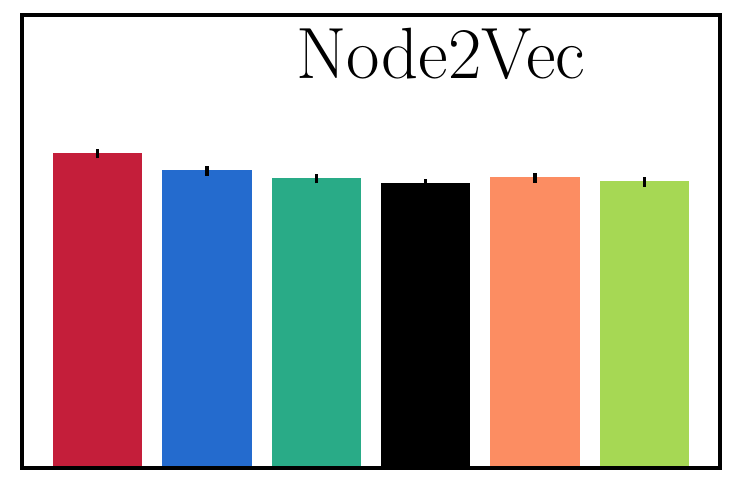}} 
\subfloat{\includegraphics[width=0.19\textwidth]{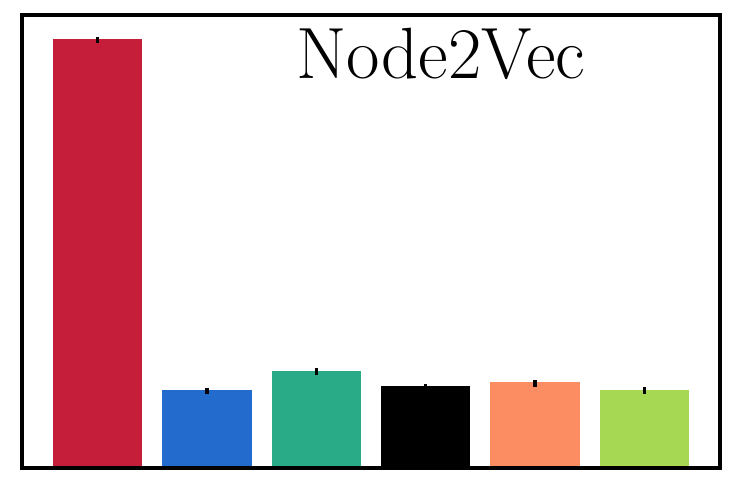}} 
\subfloat{\includegraphics[width=0.19\textwidth]{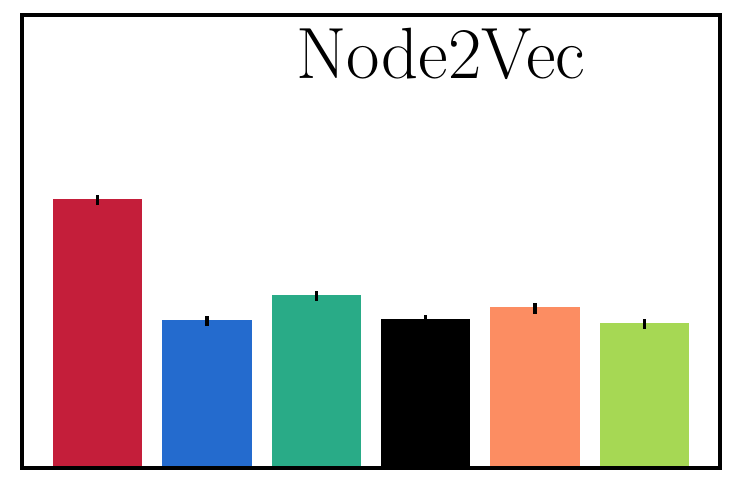}} \\
\subfloat{\includegraphics[height=0.13\textwidth]{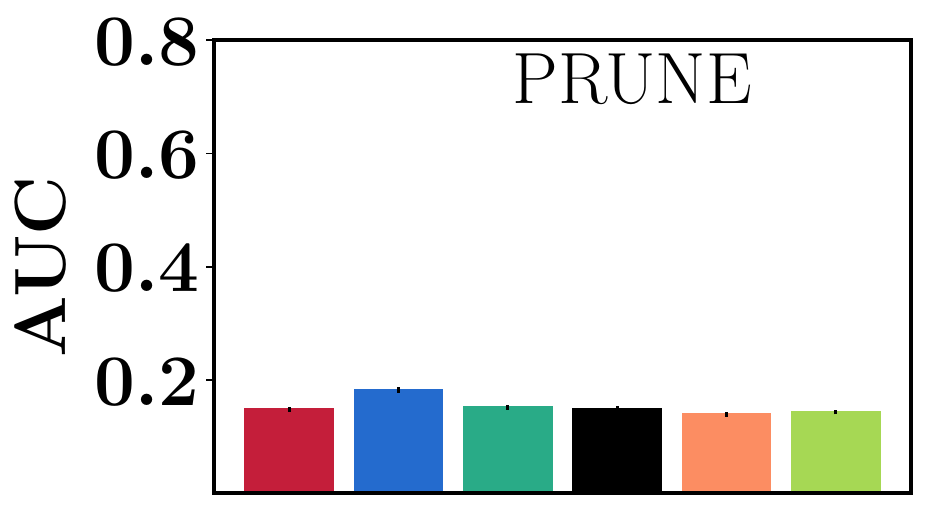}} 
\subfloat{\includegraphics[width=0.19\textwidth]{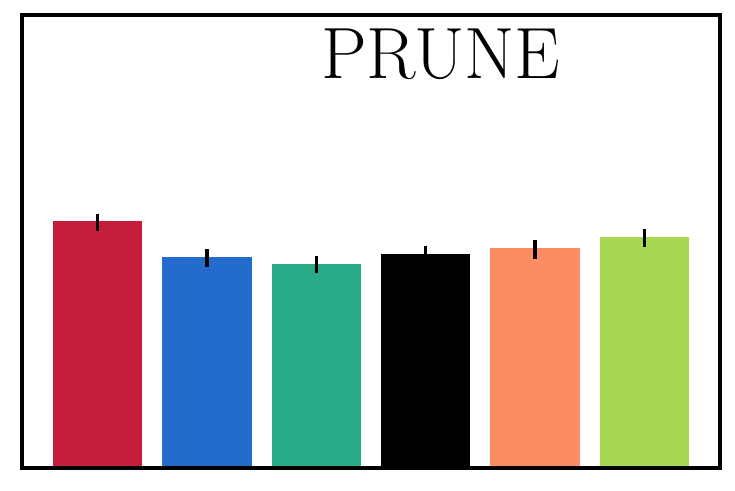}} 
\subfloat{\includegraphics[width=0.19\textwidth]{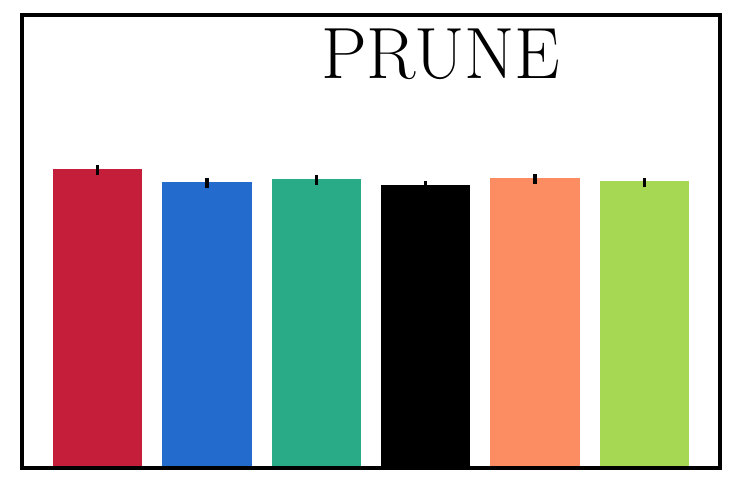}} 
\subfloat{\includegraphics[width=0.19\textwidth]{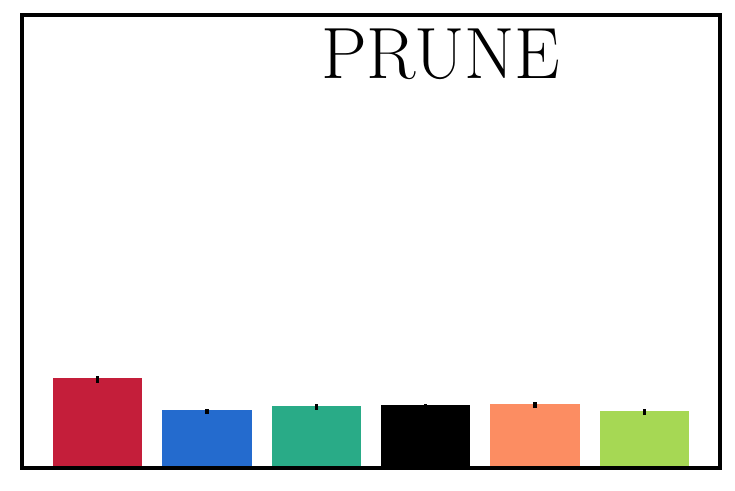}} 
\subfloat{\includegraphics[width=0.19\textwidth]{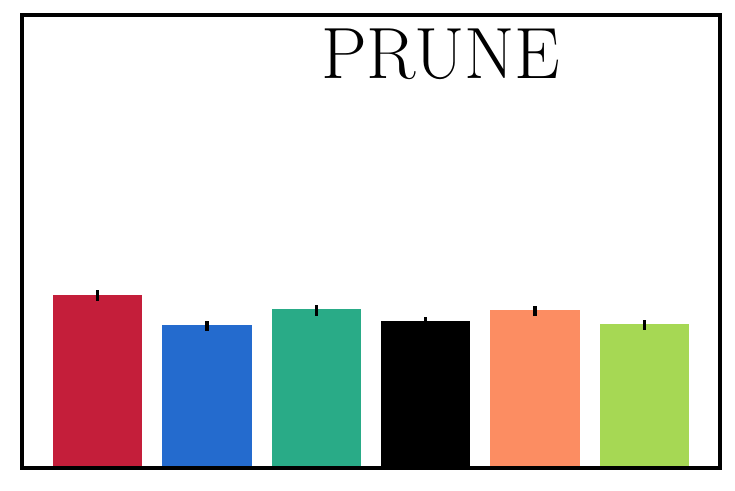}}\\
\subfloat{\includegraphics[height=0.13\textwidth]{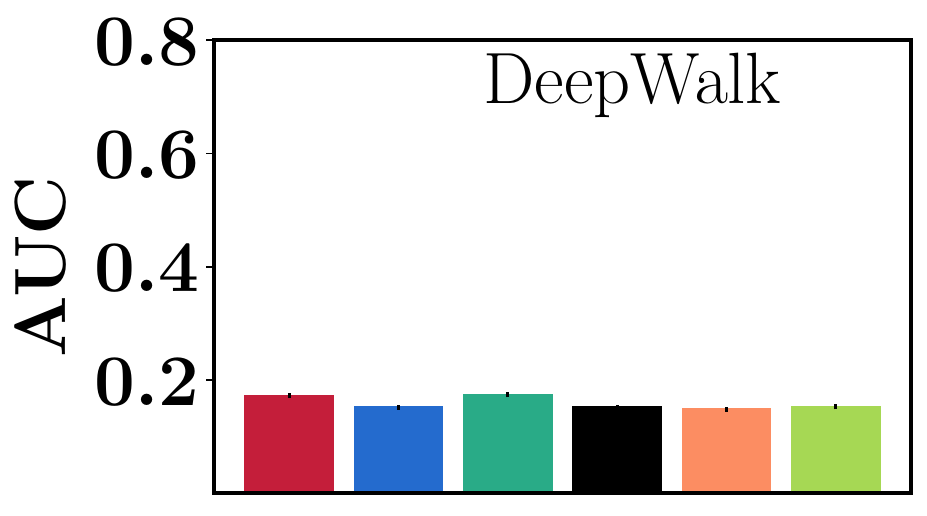}} 
\subfloat{\includegraphics[width=0.19\textwidth]{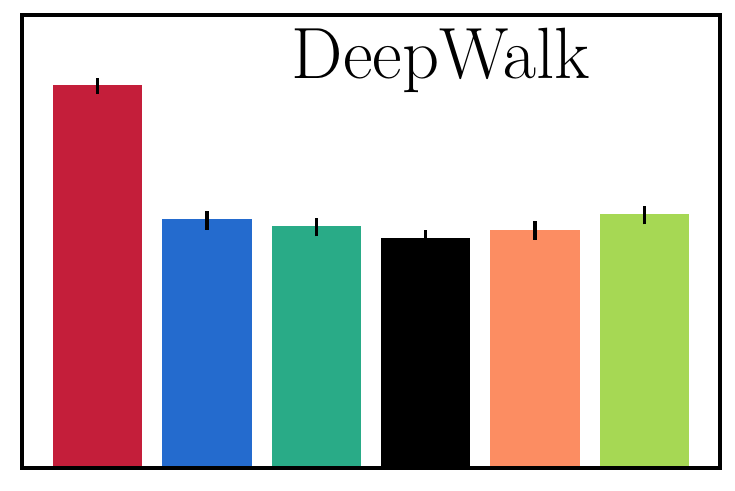}} 
\subfloat{\includegraphics[width=0.19\textwidth]{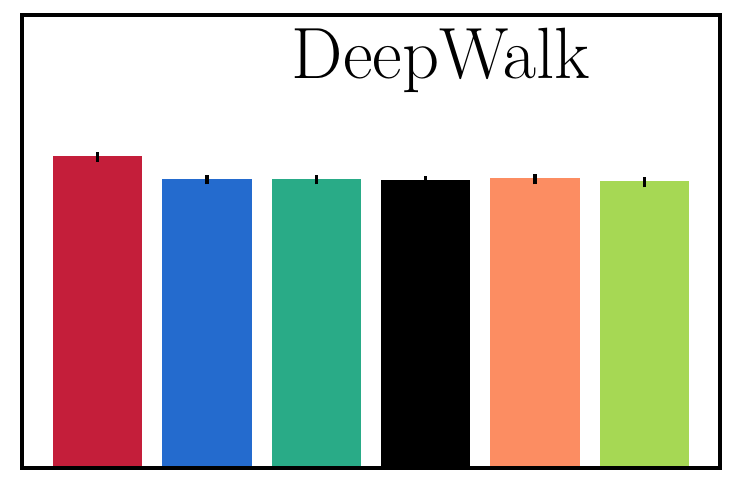}} 
\subfloat{\includegraphics[width=0.19\textwidth]{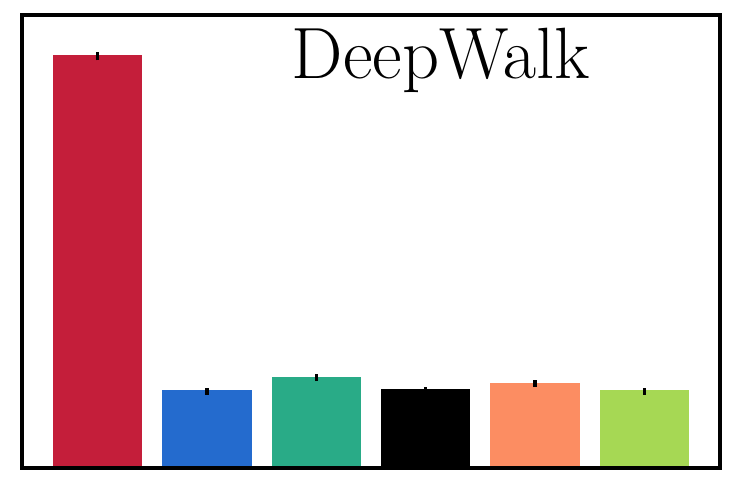}} 
\subfloat{\includegraphics[width=0.19\textwidth]{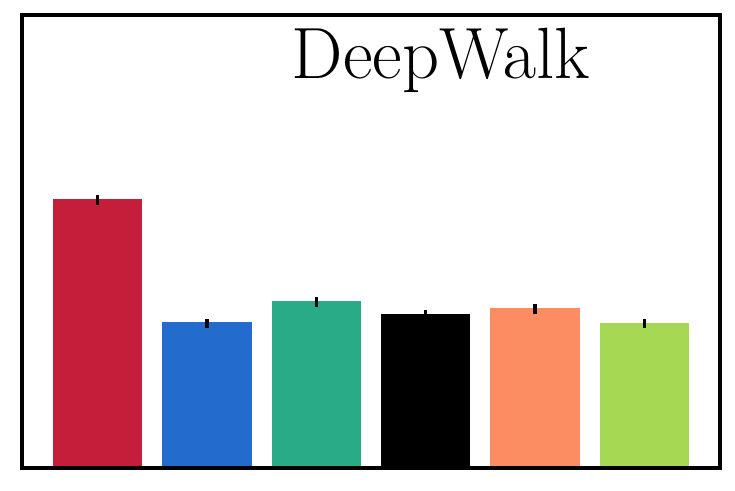}}\\
\subfloat{\includegraphics[height=0.13\textwidth]{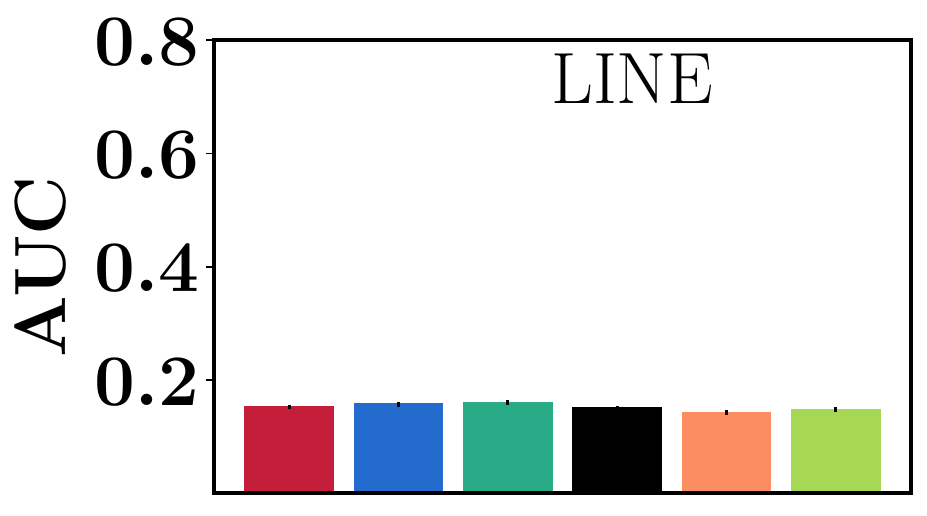}} 
\subfloat{\includegraphics[width=0.19\textwidth]{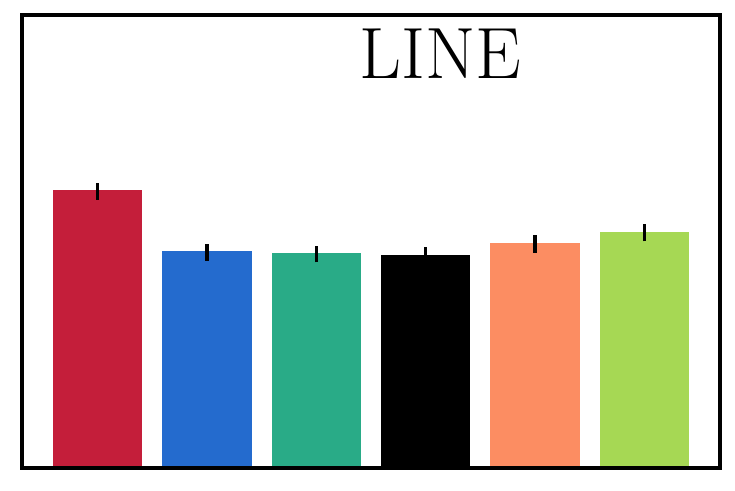}} 
\subfloat{\includegraphics[width=0.19\textwidth]{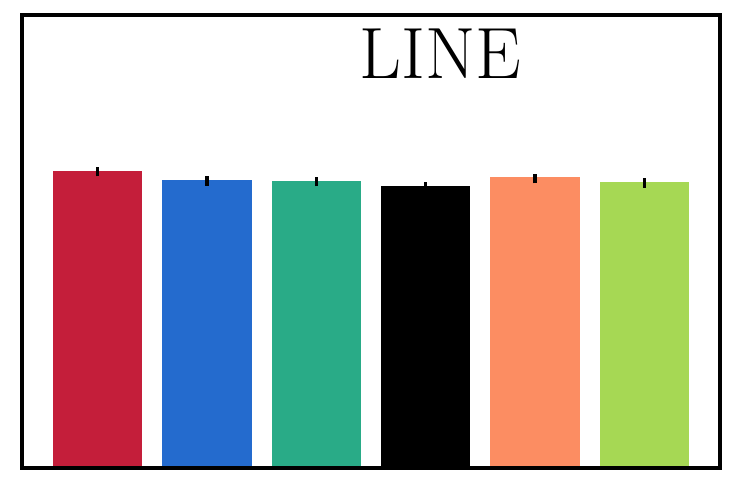}} 
\subfloat{\includegraphics[width=0.19\textwidth]{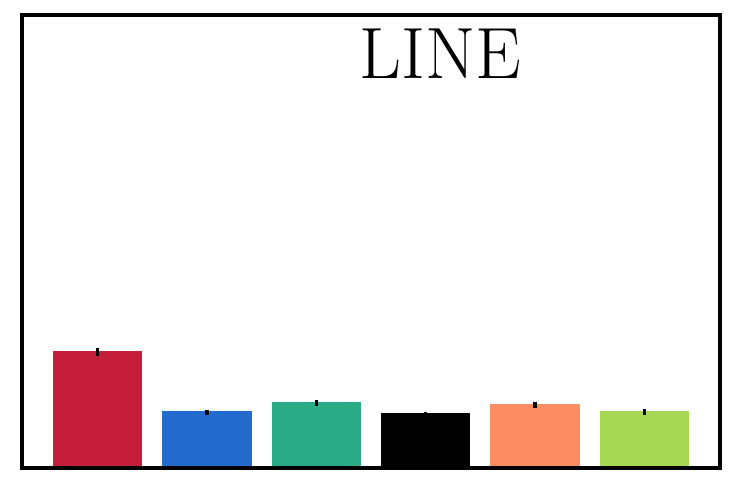}} 
\subfloat{\includegraphics[width=0.19\textwidth]{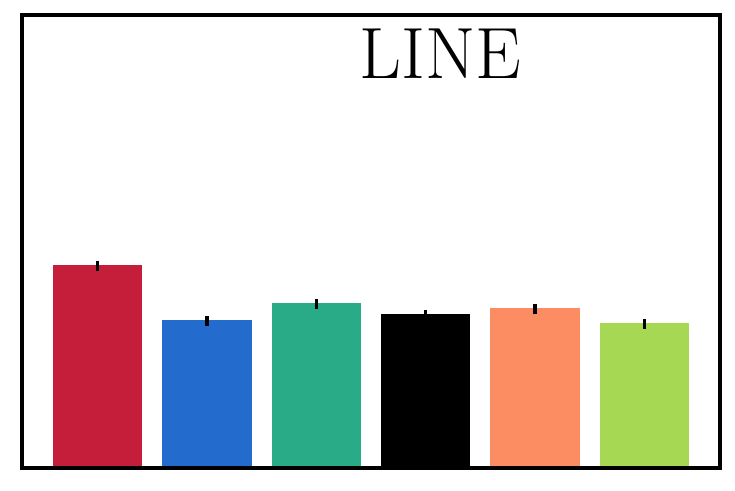}}\\
\subfloat[Facebook]{\setcounter{subfigure}{1}\includegraphics[height=0.13\textwidth]{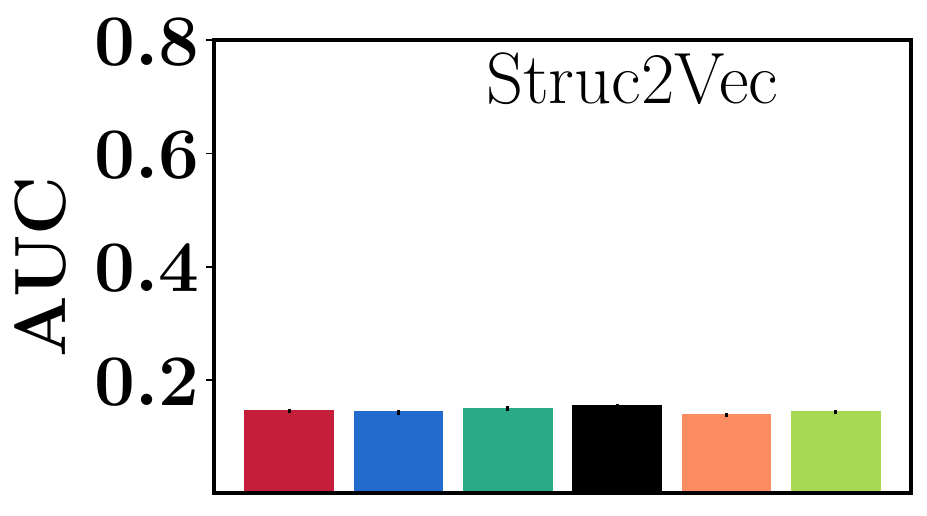}} 
\subfloat[USAir]{\includegraphics[width=0.19\textwidth]{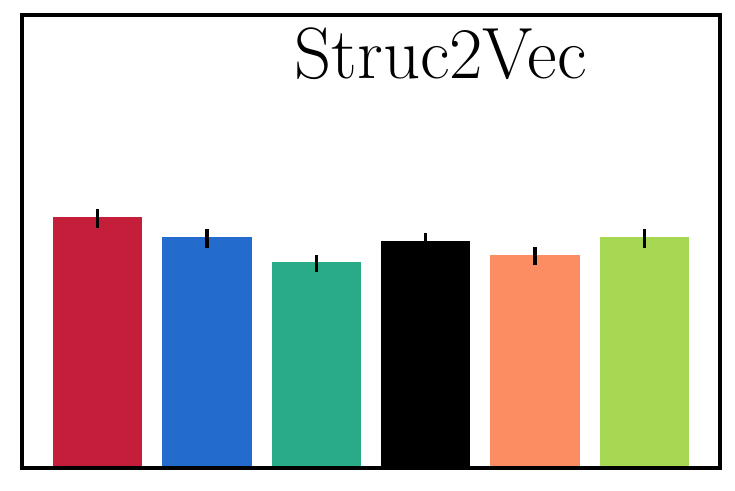}} 
\subfloat[Twitter]{\includegraphics[width=0.19\textwidth]{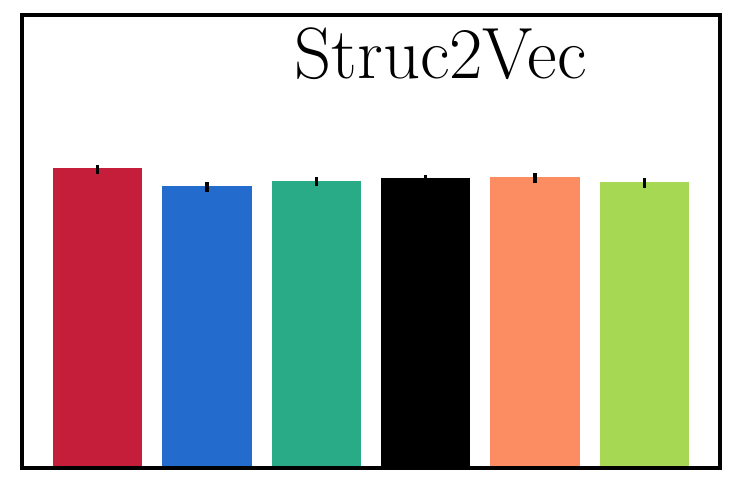}} 
\subfloat[Yeast]{\includegraphics[width=0.19\textwidth]{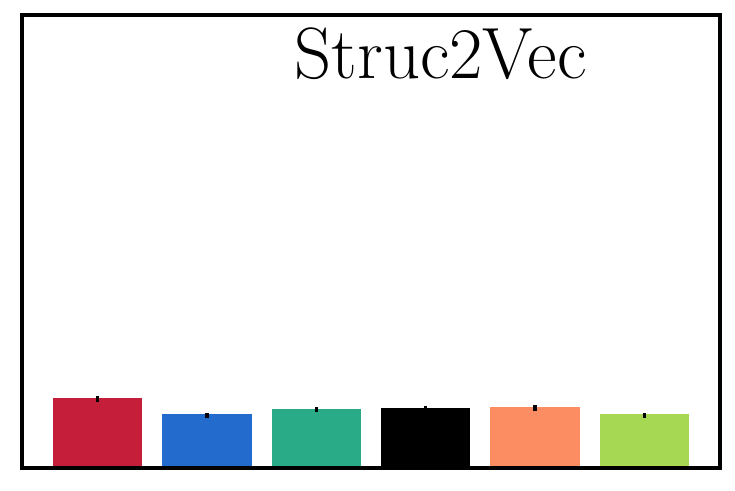}} 
\subfloat[PB]{\includegraphics[width=0.19\textwidth]{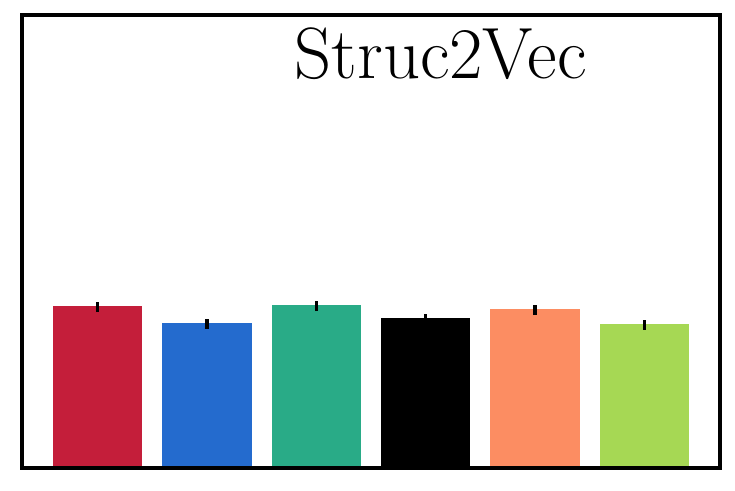}} } 
\caption{Performance (AUC) on {Facebook, USAir, Twitter, Yeast, PB datasets} for various DP algorithms: \our, \our-UMNN, \our-Lin, Staircase, Laplace and Exponential for 20\% held-out set with $\epsilon_p=0.1$ and $K=30$
for deep embedding based LP protocols (GCN, Node2Vec, PRUNE, DeepWalk, LINE, Struc2Vec).
\our\ outperforms the baselines in majority of the embedding methods. However, DP deep methods generally offer worse ranking utility than DP triad-based methods.}
\label{fig:Add2}
\end{figure*}

\end{document}